\documentclass[12pt,oneside]{amsart}

\usepackage{amsmath, amssymb, amsthm, amscd, 
}

\subjclass{37K30, 37K35}

\allowdisplaybreaks[4]

\newcommand{\eval}[2][\right]{\relax
  \ifx#1\right\relax \left.\fi#2#1\rvert}

\newcommand{\pd}{{\partial}}

\newcommand{\al}{{\alpha}}
\newcommand{\la}{{\lambda}}

\newcommand{\er}{\eqref}
\newcommand{\cl}{\colon}

\newcommand{\beq}{\begin{equation}}
\newcommand{\ee}{\end{equation}}

\newcommand{\bmu}{\begin{multline}}
\newcommand{\emul}{\end{multline}}

\newcommand{\qalg}{\mathcal{Q}}
\newcommand{\ipol}{\mathfrak{I}}

\newcommand{\frl}{\mathfrak{F}}
\newcommand{\frid}{\mathfrak{I}}
\newcommand{\fla}{\mathbf{A}}
\newcommand{\flb}{\mathbf{B}}
\newcommand{\flz}{\mathbf{Z}}

\newcommand{\ds}{\mathbf{S}}

\newcommand{\be}{{\beta}}

\newcommand{\CE}{\mathcal{E}}
\newcommand{\ce}{\mathcal{E}}

\newcommand{\mat}{\mathcal{M}}

\newcommand{\zp}{\mathbb{Z}_{\ge 0}}
\newcommand{\zsp}{\mathbb{Z}_{>0}}

\renewcommand{\sl}{\mathfrak{sl}}

\newcommand{\mg}{\mathfrak{g}}

\newcommand{\bl}{\mathfrak{L}}

\newcommand{\ga}{\mathbb{A}}
\newcommand{\gb}{\mathbb{B}}

\newcommand{\mR}{\mathfrak{R}}

\newcommand{\lb}{\label}

\newcommand{\vf}{\varphi}

\newcommand{\Com}{\mathbb{C}}

\newcommand{\ve}{\mathcal{D}}

\newcommand{\fd}{\mathbb{F}}
\newcommand{\fdn}{\mathbb{F}}

\newcommand{\cprime}{\/{\mathsurround=0pt$'$}}

\newcommand{\hkneq}{\mathrm{KN}}

\newcommand{\knpr}{v}
\newcommand{\oc}{p}
\newcommand{\ocn}{p}
\newcommand{\dc}{q}
\newcommand{\nv}{m}
\newcommand{\eo}{d}

\newtheorem{theorem}{Theorem}
\newtheorem{proposition}{Proposition}
\newtheorem{lemma}{Lemma}

\theoremstyle{definition}

\newtheorem{example}{Example}
\newtheorem{remark}{Remark}

\begin{document}


\author{Sergey Igonin}
\address{Department of Mathematics \\ 
Utrecht University \\
P.O. Box 80010 \\ 
3508 TA Utrecht \\ 
the Netherlands} 

\email{s-igonin@yandex.ru}

\title[Higher jet prolongation Lie algebras and B\"acklund transformations]
{Higher jet prolongation Lie algebras 
and B\"acklund transformations for $(1+1)$-dimensional PDEs}


\begin{abstract}
For any $(1+1)$-dimensional (multicomponent) evolution PDE, 
we define a sequence of Lie algebras $\fdn^\ocn$, $\ocn=0,1,2,3,\dots$, 
which are responsible for all Lax pairs and zero-curvature representations (ZCRs) 
of this PDE. 

In our construction, jets of arbitrary order are allowed. 
In the case of lower order jets, the algebras $\fdn^\ocn$ generalize 
Wahlquist-Estabrook prolongation algebras. 

To achieve this, we find a normal form for (nonlinear) ZCRs with respect 
to the action of the group of gauge transformations. 
One shows that 
any ZCR is locally gauge equivalent to the ZCR arising from a vector field representation of the algebra $\fdn^\ocn$, where $\ocn$ is the order of jets 
involved in the $x$-part of the ZCR.

More precisely, we define a Lie algebra $\fdn^\ocn$ for each nonnegative integer~$\ocn$ and each point~$a$ of the infinite prolongation~$\CE$ of the evolution PDE. 
So the full notation for the algebra is $\fdn^\ocn(\CE,a)$.  

Using these algebras, 
one obtains a necessary condition for two given evolution PDEs  
to be connected by a B\"acklund transformation. 

In this paper, the algebras $\fdn^\ocn(\CE,a)$ are computed for some PDEs 
of KdV type. 
In a different paper with G.~Manno, 
we compute $\fdn^\ocn(\CE,a)$ 
for multicomponent Landau-Lifshitz systems of Golubchik and Sokolov.  
Among the obtained Lie algebras, one encounters  
infinite-dimensional algebras of certain matrix-valued functions on some algebraic curves. 
Besides, some solvable ideals and semisimple Lie algebras appear in the description 
of $\fdn^\ocn(\CE,a)$.

Applications to classification of KdV and Krichever-Novikov type equations 
with respect to B\"acklund transformations are also briefly discussed.

\end{abstract}


\maketitle

\tableofcontents

\newpage

\section{Introduction}

\subsection{The main results}

\lb{main_int}



A large part of the 
theory of integrable systems is devoted to $(1+1)$-dimensional evolution PDEs 
\begin{gather}
\label{sys_intr}
\frac{\pd u^i}{\pd t}
=F^i(x,t,u^1,\dots,u^\nv,\,u^1_1,\dots,u^\nv_1,\dots,u^1_{\eo},\dots,u^\nv_{\eo}),\\
\notag
u^i=u^i(x,t),\quad\qquad u^i_k=\frac{\pd^k u^i}{\pd x^k},\quad\qquad  
i=1,\dots,\nv,\quad\qquad k\in\zsp. 
\end{gather}
Here the number $\eo$ is such that the functions $F^i$ may depend only 
on the variables $x$, $t$, $u^j$, $u^j_k$ for $k\le\eo$.

This class of PDEs includes many celebrated equations of mathematical physics 
(e.g., the KdV, Landau-Lifshitz, nonlinear Schr\"odinger equations). 

Many more PDEs can be written in the evolution form~\er{sys_intr} after a suitable change of variables\footnote{It is known that almost any determined system of PDEs 
with two independent variables can be written in the evolution form~\er{sys_intr} 
by means of a change of variables.}.
For example, the sine-Gordon equation $u_{tt}-u_{xx}=\sin u$ is equivalent to the evolution system 
$$
u^1_t=u^2,\qquad\qquad u^2_t=u^1_{xx}+\sin u^1,
$$
where $u^1=u$, $u^2=u_t$, and subscripts denote derivatives.

In this paper, integrability of PDEs is understood 
in the sense of soliton theory and the inverse scattering method. 
This is sometimes called $S$-integrability. 

It is well known that, in order to understand 
integrability properties of~\er{sys_intr}, one needs to study overdetermined systems of the form
\begin{gather}
\lb{intcov}
\begin{aligned}
w^j_x&=
\al^j(w^1,\dots,w^\dc,x,t,u^1,\dots,u^\nv,u^1_1,\dots,u^\nv_1,\dots,u^1_\ocn,\dots,u^\nv_\ocn),\\
w^j_t&=
\beta^j(w^1,\dots,w^\dc,x,t,u^1,\dots,u^\nv,\,u^1_1,\dots,u^\nv_1,\dots,u^1_{\ocn+\eo-1},\dots,u^\nv_{\ocn+\eo-1}),
\end{aligned}\\
\notag
w^j=w^j(x,t),\qquad\qquad j=1,\dots,\dc, 
\end{gather}
such that system~\eqref{intcov} is compatible modulo~\eqref{sys_intr}.
The precise meaning of this compatibility condition is explained in Remark~\ref{compatcond} below.

It is well known that Lax pairs, B\"acklund transformations, 
and zero-curvature representations for~\er{sys_intr} can be described in terms of systems~\er{intcov} 
compatible modulo~\er{sys_intr}. 
Thus compatible systems~\er{intcov} are of fundamental importance for the theory of nonlinear PDEs 
in two independent variables $x$, $t$. 

Set $u^i_0=u^i$. 
The number $\ocn$ in~\er{intcov} is such that  
the functions $\al^j$ may depend only on the variables $w^l$, $x$, $t$, $u^i_{k}$ for 
$0\le k\le\ocn$.
Then, as is explained in Remark~\ref{compatcond}, 
the compatibility condition implies that the functions $\beta^j$ may depend only on 
$w^l$, $x$, $t$, $u^{i'}_{k'}$ for $0\le k'\le\ocn+\eo-1$.

If the functions $\al^j$, $\be^j$ are linear with respect to $w^1,\dots,w^\dc$, 
then \er{intcov} corresponds to a zero-curvature representation for system~\er{sys_intr}. 
In the case of nonlinear functions $\al^j$, $\be^j$, 
a compatible system~\er{intcov} can be regarded as a nonlinear zero-curvature representation for~\er{sys_intr}. 

In this paper, we study the following problem. 
Given a system~\er{sys_intr}, how to describe all systems~\er{intcov} 
that are compatible modulo~\er{sys_intr}? 

In the case when $\ocn=0$ and the functions $F^i$, $\al^j$, $\beta^j$ do not depend on $x$, $t$, 
a partial answer to this question is provided by the Wahlquist-Estabrook prolongation
method (WE method for short). Namely, for a given system~\er{sys_intr}, 
the WE method constructs a Lie algebra in terms of generators and relations such that 
compatible systems of the form  
\begin{gather}
\lb{wecov}
\begin{aligned}
w^j_x&=
\al^j(w^1,\dots,w^\dc,u^1,\dots,u^\nv),\\
w^j_t&=
\beta^j(w^1,\dots,w^\dc,u^1,\dots,u^\nv,\,u^1_1,\dots,u^\nv_1,\dots,u^1_{\eo-1},\dots,u^\nv_{\eo-1}),
\end{aligned}\\
\notag
w^j=w^j(x,t),\qquad\qquad j=1,\dots,\dc, 
\end{gather}
correspond to representations of this algebra by vector fields on the manifold $W$ 
with coordinates $w^1,\dots,w^\dc$ (see, e.g.,~\cite{dodd,nonl,Prol}) and references therein).
This algebra is called the \emph{Wahlquist-Estabrook prolongation algebra}.

In order to study the general case of systems~\er{intcov} with arbitrary $\ocn$, 
we need to consider gauge transformations. 
A \emph{gauge transformation} is given by an invertible change of variables 
\begin{equation}
\label{intgt}
  x\mapsto x,\quad t\mapsto t,\quad u^i\mapsto u^i,\quad u^i_k\mapsto u^i_k,\quad
  w^j\mapsto g^j(\tilde w^1,\dots,\tilde w^\dc,x,t,u^i,u^i_l,\dots),\quad j=1,\dots,\dc.
\end{equation}
Substituting~\er{intgt} to~\er{intcov}, we obtain equations of the form
\begin{gather}
\lb{tcov}
\begin{aligned}
\tilde w^j_x&=
\tilde\al^j(\tilde w^1,\dots,\tilde w^\dc,x,t,u^i,u^i_k,\dots),\\
\tilde w^j_t&=
\tilde\beta^j(\tilde w^1,\dots,
\tilde w^\dc,x,t,u^i,u^i_{k},\dots),
\end{aligned}\\
\notag
\tilde w^j=\tilde w^j(x,t),\qquad\qquad j=1,\dots,\dc. 
\end{gather} 
System~\er{tcov} is said to be \emph{gauge equivalent} to system~\er{intcov} if 
\er{tcov} and \er{intcov} are connected by an invertible change of variables 
of the form~\er{intgt}.

If~\er{intcov} is compatible then for any gauge transformation~\er{intgt} 
the corresponding system~\er{tcov} is compatible as well.

The WE method does not consider gauge transformations. 
In the classification of compatible systems~\er{wecov} this is acceptable, 
because the class of systems~\er{wecov} is relatively small.  

The class of systems~\er{intcov} is much larger than that of~\er{wecov}.
As we show below, gauge transformations play a very important role in the classification 
of compatible systems~\er{intcov}. 
Because of this, the classical WE method does not produce satisfactory results for~\er{intcov}.  
 
To overcome this problem, we combine 
the technique of gauge transformations with ideas similar to the WE method. 
Loosely speaking, the main results can be stated as follows.

We find a normal form for systems~\er{intcov}
with respect to the action of the group of gauge transformations.
This allows us to define a Lie algebra $\fdn^{\ocn}$ for each $\ocn\in\zp$ such 
that the following properties hold. 
Any compatible system~\er{intcov} 
is locally gauge equivalent to the system arising from a vector field representation of the algebra $\fdn^{\ocn}$. 
Two compatible systems of the form~\er{intcov} are locally gauge equivalent iff 
the corresponding vector field representations of ${\fdn}^{\ocn}$ are locally isomorphic.

More precisely, as is discussed below, we define a Lie algebra $\fdn^{\ocn}$ for each $\ocn\in\zp$ 
and each point~$a$ of the infinite prolongation $\CE$ of system~\er{sys_intr}. 
So the full notation for the algebra is $\fdn^{\ocn}(\CE,a)$.  

Recall that the \emph{infinite prolongation} $\CE$ of~\er{sys_intr} is the  
infinite-dimensional manifold with the coordinates 
\beq
\notag
x,\qquad t,\qquad u^i_k,\qquad i=1,\dots,\nv,\qquad k\in\zp.
\ee 
In this paper all manifolds, functions, vector fields, 
and maps of manifolds are supposed to be complex-analytic.
The precise definition of $\fdn^{\ocn}(\CE,a)$ for any system~\er{sys_intr} is presented in Section~\ref{csev}. 
In this definition, the algebra $\fdn^{\ocn}(\CE,a)$ is given in terms of generators and relations.

We consider representations of the Lie algebra $\fdn^{\ocn}(\CE,a)$ by vector fields on the manifold $W$ 
with coordinates $w^1,\dots,w^\dc$.
Such vector field representations of $\fdn^{\ocn}(\CE,a)$ classify 
(up to local gauge equivalence) 
all compatible systems~\er{intcov}, where functions $\al^j$, $\beta^j$ are defined on a neighborhood of 
the point $a\in\CE$. See Section~\ref{csev} for details.

Some applications of the algebras $\fdn^{\ocn}(\CE,a)$ to the theory of 
B\"acklund transformations are discussed in Subsection~\ref{subsecbt}.

According to Section~\ref{csev}, 
the algebras $\fdn^{\ocn}(\CE,a)$ for $\ocn\in\zp$ are arranged in a sequence of surjective homomorphisms 
\beq
\lb{intfdoc1}
\dots\to\fdn^{\ocn}(\CE,a)\to\fdn^{\ocn-1}(\CE,a)\to\dots\to\fdn^1(\CE,a)\to\fdn^0(\CE,a).
\ee

Let us describe the structure of $\fdn^{\ocn}(\CE,a)$ and the homomorphisms~\er{intfdoc1} more explicitly for some PDEs. 
Theorem~\ref{kntkdvtypeth} is proved in Section~\ref{seckdvtype}.
\begin{theorem}[Section~\ref{seckdvtype}]
\lb{kntkdvtypeth}
Let $\CE$ be the 
infinite prolongation of an equation of the form 
\beq 
\lb{intkdveq}
u_t=u_{xxx}+f(u,u_x),\qquad\qquad u=u(x,t),
\ee 
where $f$ is an arbitrary function. Let $a\in\CE$.
 
For each $\ocn\in\zsp$, 
consider the homomorphism $\vf_\ocn\cl\fdn^\ocn(\CE,a)\to\fdn^{\ocn-1}(\CE,a)$ from~\er{intfdoc1}. 
Then we have 
\beq
\notag
[v_1,v_2]=0\qquad\qquad\forall\,v_1\in\ker\vf_\ocn,\qquad\forall\,v_2\in\fdn^\ocn(\CE,a).
\ee
That is, the kernel of $\vf_\ocn$ is contained in the center of the Lie 
algebra $\fdn^\ocn(\CE,a)$.

For each $k\in\zsp$, let $\psi_k\colon\fdn^k(\CE,a)\to\fdn^{0}(\CE,a)$ 
be the composition of the homomorphisms
\beq
\notag
\fdn^k(\CE,a)\to\fdn^{k-1}(\CE,a)\to\dots
\to\fdn^{1}(\CE,a)\to\fdn^{0}(\CE,a)
\ee
from~\er{intfdoc1}. Then   
\beq
\notag
[h_1,[h_2,\dots,[h_{k-1},[h_k,h_{k+1}]]\dots]]=0\qquad\qquad\forall\,h_1,\dots,h_{k+1}\in\ker\psi_k.
\ee
In particular, the kernel of $\psi_k$ is nilpotent.
\end{theorem}

Let $\CE$ be the infinite prolongation of the KdV equation 
\beq
\lb{intkdv}
u_t=u_{xxx}+u_xu.
\ee
Consider the infinite-dimensional Lie algebra 
$\sl_2(\Com[\la])\cong \sl_2(\mathbb{C})\otimes_{\mathbb{C}}\Com[\lambda]$, 
where $\Com[\lambda]$ is the algebra of polynomials in $\la$.

It is shown in~\cite{lie-scal} that, for the KdV equation, 
the algebra $\fdn^0(\CE,a)$ is isomorphic to the direct sum of $\sl_2(\Com[\la])$ and 
a $3$-dimensional abelian Lie algebra. 
Combining this with Theorem~\ref{kntkdvtypeth}, we obtain the following.

\begin{theorem}
\lb{intthkdv}
Let $\CE$ be the infinite prolongation of the KdV equation~\er{intkdv}.
Let $a\in\CE$. 
Then 
\begin{itemize}
 \item the algebra $\fdn^0(\CE,a)$ is isomorphic to the direct sum of $\sl_2(\Com[\la])$ and 
a $3$-dimensional abelian Lie algebra,
\item for each $\ocn\in\zsp$, the kernel of the surjective homomorphism $\fdn^\ocn(\CE,a)\to\fdn^{0}(\CE,a)$ from~\er{intfdoc1} is nilpotent.  
\end{itemize}
\end{theorem}

To describe $\fdn^0(\CE,a)$ for the KdV equation, the paper~\cite{lie-scal} uses 
the following fact. For the KdV equation (and some other PDEs), the algebra $\fdn^0(\CE,a)$ 
is isomorphic to a certain subalgebra of the Wahlquist-Estabrook prolongation algebra. 
The explicit structure of the Wahlquist-Estabrook prolongation algebra for the KdV equation 
is given in~\cite{kdv,kdv1}, and this allows us to describe $\fdn^0(\CE,a)$ 
(see~\cite{lie-scal} for details).

\begin{remark}
Using some extra computations, one can prove the following.
\begin{proposition}
\lb{kdvnilp}
Let $\CE$ be the infinite prolongation of the KdV equation~\er{intkdv}. 
For any $a\in\CE$ and any $\oc\in\zp$, the algebra $\fdn^\ocn(\CE,a)$
is isomorphic to the direct sum of $\sl_2(\Com[\la])$ and a finite-dimensional nilpotent Lie algebra.
\end{proposition}
We do not present the proof of Proposition~\ref{kdvnilp} in this paper, 
because the explicit structure of nilpotent ideals of $\fdn^\ocn(\CE,a)$ 
is not needed for the main applications to B\"acklund transformations. 
See Remark~\ref{nsideals} below for a discussion of this.
\end{remark}

For any constants $e_1,e_2,e_3\in\Com$, 
consider the Krichever-Novikov equation~\cite{krich80,svin-sok83}
\begin{equation}
\lb{kneqprop}
  \hkneq(e_1,e_2,e_3)=\left\{
 u_t=u_{xxx}-\frac32\frac{u_{xx}^2}{u_x}+
 \frac{(u-e_1)(u-e_2)(u-e_3)}{u_x},\ \quad u=u(x,t)\right\}.
\end{equation}
To study the algebras $\fdn^\ocn(\CE,a)$ for this equation, 
we need some auxiliary constructions.

Let $\mathbb{C}[v_1,v_2,v_3]$ be the algebra of polynomials 
in the variables $v_1$, $v_2$, $v_3$. 
Let $e_1,e_2,e_3\in\Com$ be such that $e_1\neq e_2\neq e_3\neq e_1$. 
Consider the ideal $\mathcal{I}_{e_1,e_2,e_3}\subset\mathbb{C}[v_1,v_2,v_3]$ generated by the
polynomials
\begin{equation}
  \label{elc}
  v_i^2-v_j^2+e_i-e_j,\qquad\qquad i,\,j=1,2,3.
\end{equation}

Set 
$$
E_{e_1,e_2,e_3}=\mathbb{C}[v_1,v_2,v_3]/\mathcal{I}_{e_1,e_2,e_3}.
$$ 
In other words, $E_{e_1,e_2,e_3}$ is the commutative associative algebra of regular 
functions on the algebraic curve 
in $\mathbb{C}^3$ defined by the polynomials~\eqref{elc}.
It is easy to check that this curve is nonsingular and is of genus~$1$. 

We have the natural homomorphism $\mathbb{C}[v_1,v_2,v_3]\to E_{e_1,e_2,e_3}$.
The image of $v_i\in\mathbb{C}[v_1,v_2,v_3]$ in~$E_{e_1,e_2,e_3}$ 
is denoted by $\bar v_i\in E_{e_1,e_2,e_3}$ for $i=1,2,3$.

Consider also a basis $x_1$, $x_2$, $x_3$ of the Lie algebra
$\mathfrak{so}_3(\mathbb{C})$ such that 
\begin{equation*}
[x_1,x_2]=x_3,\qquad [x_2,x_3]=x_1,\qquad [x_3,x_1]=x_2. 
\end{equation*}
We endow the space $\mathfrak{so}_3(\mathbb{C})\otimes_\mathbb{C} E_{e_1,e_2,e_3}$ with 
the following Lie algebra structure 
$$
[y_1\otimes h_1,\,y_2\otimes h_2]=[y_1,y_2]\otimes h_1h_2,\qquad\quad 
y_1,y_2\in\mathfrak{so}_3(\mathbb{C}),\qquad\quad h_1,h_2\in E_{e_1,e_2,e_3}.
$$

Denote by $\mR_{e_1,e_2,e_3}$ the Lie subalgebra of $\mathfrak{so}_3(\mathbb{C})\otimes_\mathbb{C} E_{e_1,e_2,e_3}$ generated by the elements
$$
x_i\otimes\bar v_i\,\in\,\mathfrak{so}_3(\mathbb{C})\otimes_\mathbb{C} E_{e_1,e_2,e_3},\qquad\qquad i=1,2,3.
$$
It is easily seen that the Lie algebra $\mR_{e_1,e_2,e_3}$ is infinite-dimensional. 
According to~\cite{ll}, the Wahlquist-Estabrook prolongation algebra of the 
anisotropic Landau-Lifshitz equation is isomorphic to the 
direct sum of $\mR_{e_1,e_2,e_3}$ and a $2$-dimensional abelian Lie algebra.

According to Proposition~\ref{fdockn} below, the algebra $\mR_{e_1,e_2,e_3}$ 
appears also in the structure of the algebras $\fdn^\ocn(\CE,a)$ 
for the Krichever-Novikov equation~\er{kneqprop}.
A proof of Proposition~\ref{fdockn} is sketched in~\cite{cfg}.
\begin{proposition}[\cite{cfg}]
\lb{fdockn}
For any constants $e_1,e_2,e_3\in\Com$, 
consider the Krichever-Novikov equation $\hkneq(e_1,e_2,e_3)$ given by~\er{kneqprop}.
Let $\CE$ be the infinite prolongation of this equation. Let $a\in\CE$.
Then 
\begin{itemize}
 \item the algebra $\fdn^0(\CE,a)$ is zero,
\item for any $\ocn\ge 2$, the kernel of the surjective homomorphism $\fdn^{\ocn}(\CE,a)\to\fdn^{1}(\CE,a)$ 
from~\er{intfdoc1} is nilpotent,
\item if $e_1\neq e_2\neq e_3\neq e_1$, then $\fdn^1(\CE,a)$ is isomorphic to $\mR_{e_1,e_2,e_3}$.
\end{itemize}
\end{proposition}
\begin{remark}
The proof of Proposition~\ref{fdockn} uses the well-known fact that 
the Krichever-Novikov equation~\er{kneqprop} possesses an $\mathfrak{so}_3$-valued 
zero-curvature representation parametrized by the above-mentioned curve.
\end{remark}
\begin{remark}
As has been said above, 
for some evolution PDEs the algebra $\fdn^{0}(\CE,a)$ is isomorphic to a subalgebra of 
the Wahlquist-Estabrook prolongation algebra. 

The algebras $\fdn^{\ocn}(\CE,a)$ for $\ocn\ge 1$ cannot be obtained by the classical Wahlquist-Estabrook prolongation method, because the main idea behind 
the definition of $\fdn^{\ocn}(\CE,a)$  is based on the use of gauge transformations,  
while the Wahlquist-Estabrook prolongation method does not consider gauge transformations.  

According to Proposition~\ref{fdockn}, 
for the Krichever-Novikov equation~\er{kneqprop} 
we have $\fdn^{0}(\CE,a)=0$ and $\dim\fdn^{\ocn}(\CE,a)=\infty$ for $\ocn\ge 1$ 
(in the case $e_1\neq e_2\neq e_3\neq e_1$). 
It is easy to show that the classical Wahlquist-Estabrook prolongation algebra is trivial 
for the Krichever-Novikov equation.  
Thus in this example the algebras $\fdn^{\ocn}(\CE,a)$ are much more interesting 
than the Wahlquist-Estabrook prolongation algebra.   
\end{remark}

As another example, 
we consider a multicomponent generalization of the Landau-Lifshitz equation 
from~\cite{mll,skr-jmp}.  
To present this PDE, we need some notation.
Fix an integer $n\ge 3$. 
For any $n$-dimensional vectors 
$V={(v^1,\dots,v^n)}$ and $Y={(y^1,\dots,y^n)}$, set 
$\langle V,Y\rangle=\sum_{i=1}^nv^iy^i$.

Let $r_1,\dots,r_n\in\Com$ be such that $r_i\neq r_j$ for all $i\neq j$. 
Denote by $R=\mathrm{diag}\,(r_1,\dots,r_n)$ the diagonal $(n\times n)$-matrix 
with entries $r_i$.
Consider the PDE
\begin{equation}
\label{mainint}
S_t=\Big(S_{xx}+\frac32\langle S_x,S_x\rangle S\Big)_x+\frac32\langle S,RS\rangle S_x,
\qquad\,\quad \langle S,S\rangle=1,\qquad\quad
R=\mathrm{diag}\,(r_1,\dots,r_n),
\end{equation} 
where $S=\big(s^1(x,t),\dots,s^n(x,t)\big)$ 
is a column-vector of dimension~$n$, and $s^i(x,t)$ take values in $\Com$.

System~\eqref{mainint} was introduced in~\cite{mll}. 
According to~\cite{mll}, 
for $n=3$ it coincides with the higher symmetry (the commuting flow) 
of third order for the Landau-Lifshitz equation. 
Thus~\eqref{mainint} can be regarded as an $n$-component generalization of the Landau-Lifshitz equation. 

The paper~\cite{mll} considers also the following algebraic curve 
\begin{equation}
\label{curve}
\la_i^2-\la_j^2=r_j-r_i,\qquad\qquad i,j=1,\dots,n,
\end{equation}
in the space $\Com^n$ with coordinates $\la_1,\dots,\la_n$.
According to~\cite{mll}, this curve is of genus ${1+(n-3)2^{n-2}}$, 
and system~\eqref{mainint} possesses a zero-curvature representation (Lax pair) 
parametrized by points of this curve. 

System~\eqref{mainint} has an infinite number of symmetries, 
conservation laws~\cite{mll}, and an auto-B\"acklund transformation~\cite{ll-backl}. 
Soliton-like solutions of~\eqref{mainint} are presented in~\cite{ll-backl}. 
In~\cite{skr-jmp} system~\eqref{mainint}
and its symmetries are constructed by means of the Kostant--Adler scheme.



Denote by $\mathfrak{gl}_{n+1}(\Com)$ the space of matrices of 
size $(n+1)\times(n+1)$ with entries from~$\Com$.
Let $E_{i,j}\in\mathfrak{gl}_{n+1}(\Com)$ be the matrix with 
$(i,j)$-th entry equal to 1 and all other entries equal to zero. 

Let $\mathfrak{so}_{n,1}\subset\mathfrak{gl}_{n+1}(\Com)$ be 
the Lie algebra of the matrix Lie group $\mathrm{O}(n,1)$, 
which consists of linear transformations that preserve the standard 
bilinear form of signature~$(n,1)$.
The algebra $\mathfrak{so}_{n,1}$ has the following basis
$$
E_{i,j}-E_{j,i},\qquad i<j\le n,\qquad\qquad 
E_{l,n+1}+E_{n+1,l},\qquad l=1,\dots,n.
$$
We regard $\la_1,\dots,\la_n$ as abstract variables 
and consider the algebra $\Com[\la_1,\dots,\la_n]$ 
of polynomials in $\la_1,\dots,\la_n$. 
Let $\ipol\subset\Com[\la_1,\dots,\la_n]$ be the ideal  
generated by $\la_i^2-\la_j^2+r_i-r_j$ for $i,j=1,\dots,n$. 
 
Consider the quotient algebra $\qalg=\Com[\la_1,\dots,\la_n]/\ipol$, which 
is isomorphic to the 
algebra of polynomial functions on the algebraic curve~\er{curve}. 

The space $\mathfrak{so}_{n,1}\otimes_\Com\qalg$ 
is an infinite-dimensional Lie algebra over $\Com$ 
with the Lie bracket 
$$
[M_1\otimes h_1,\,M_2\otimes h_2]=[M_1,M_2]\otimes h_1h_2,
\qquad\qquad 
M_1,M_2\in \mathfrak{so}_{n,1},\qquad\qquad h_1,h_2\in \qalg.
$$
We have the natural homomorphism 
$\xi\cl\Com[\la_1,\dots,\la_n]\to\Com[\la_1,\dots,\la_n]/\ipol=\qalg$. 
Set $\hat\la_i=\xi(\la_i)\in \qalg$. 

Consider the following elements of ${\mathfrak{so}_{n,1}\otimes \qalg}$
\begin{equation}
\notag
Q_i=(E_{i,n+1}+E_{n+1,i})\otimes\hat\la_i,
\qquad\qquad i=1,\dots,n.
\end{equation}
Denote by $L(n)\subset\mathfrak{so}_{n,1}\otimes \qalg$ the Lie subalgebra  
generated by~$Q_1,\dots,Q_n$. 

Since $\hat{\la}_i^2-\hat\la_j^2+r_i-r_j=0$ in $\qalg$,
the element 
$\hat\la=\hat\la_i^2+r_i\in \qalg$ does not depend on $i$.  

For $i,j\in\{1,\dots,n\}$ and $k\in\zsp$, 
consider the following elements of 
${\mathfrak{so}_{n,1}\otimes_\Com \qalg}$  
$$
Q^{2k-1}_i=(E_{i,n+1}+E_{n+1,i})\otimes\hat\la^{k-1}\hat\la_i,\qquad\qquad
Q^{2k}_{ij}=(E_{i,j}-E_{j,i})\otimes\hat\la^{k-1}\hat\la_i\hat\la_j. 
$$
According to~\cite{mll-2012}, the elements 
\beq
\notag
Q^{2k-1}_l,\quad\qquad Q^{2k}_{ij},\quad\qquad i,j,l\in\{1,\dots,n\},\quad\qquad 
i<j,\quad\qquad k\in\zsp,
\ee
form a basis of $L(n)$. 
Note that the algebra $L(n)$ is very similar to  
Lie algebras that were studied in~\cite{mll,skr,skr-jmp}. 


The following result is presented in~\cite{gll-2011}.
\begin{proposition}[\cite{gll-2011}]
Let $\CE$ be the infinite prolongation of system~\er{mainint} for $n\ge 3$. 
Let $a\in\CE$.

The Lie algebras $\fd^\oc(\CE,a)$ have the following structure.

The algebra $\fd^0(\CE,a)$ is isomorphic to $L(n)$.

There is a solvable ideal $\mathcal{I}$ of $\fd^{1}(\CE,a)$ 
such that $\fd^{1}(\CE,a)/\mathcal{I}\cong L(n)\oplus\mathfrak{so}_{n-1}$, 
where $\mathfrak{so}_{n-1}$ is the Lie algebra 
of skew-symmetric $(n-1)\times(n-1)$ matrices. 
The homomorphism $\fd^{1}(\CE,a)\to\fd^0(\CE,a)$ from~\er{intfdoc1} coincides 
with the composition 
\beq
\lb{fd1lnso}
\fd^{1}(\CE,a)\to\fd^{1}(\CE,a)/\mathcal{I}\cong L(n)\oplus\mathfrak{so}_{n-1}
\to L(n)\cong \fd^0(\CE,a). 
\ee

For any $k\ge 2$, 
the kernel of the homomorphism $\fd^{k}(\CE,a)\to\fd^{k-1}(\CE,a)$ from~\er{intfdoc1} is solvable. 

For each $k\ge 1$ there is a surjective homomorphism  
\beq
\lb{vfkfdk}
\mu_k\colon\fd^k(\CE,a)\,\to\,L(n)\oplus\mathfrak{so}_{n-1} 
\ee
such that the kernel of $\mu_k$ is solvable. 

\end{proposition}

The algebra $L(n)$ in~\er{fd1lnso},~\er{vfkfdk} essentially comes from 
the zero-curvature representation (Lax pair) for system~\er{mainint} constructed 
in~\cite{mll,skr-jmp}. The algebra $\mathfrak{so}_{n-1}$ in~\er{fd1lnso},~\er{vfkfdk} 
does not come from this zero-curvature representation.

Several more results on the structure of the algebras $\fdn^{\ocn}(\CE,a)$ for scalar 
evolution equations are described in~\cite{lie-scal}.

\begin{remark}
\lb{nsideals}
As has been said above, for equations~\er{intkdveq}, \er{intkdv}, \er{kneqprop}, \er{mainint} 
the algebras $\fdn^{\ocn}(\CE,a)$ contain some nilpotent or solvable ideals. 
The explicit structure of these ideals is not completely clear.

For the main applications to B\"acklund transformations, it is sufficient to know that these ideals are solvable.
For example, in Subsection~\ref{subsecbt} we discuss 
Proposition~\ref{knprop} about B\"acklund transformations, which is proved in~\cite{mpiprep10}.
The proof of this result in~\cite{mpiprep10} uses the quotient algebras $\fdn^{\ocn}(\CE,a)/\mathcal{S}$, 
where $\mathcal{S}$ is the sum of all solvable ideals of $\fdn^{\ocn}(\CE,a)$. 
So the explicit structure of solvable ideals of $\fdn^{\ocn}(\CE,a)$ is not needed for such results.
 
\end{remark}

\begin{remark}
\lb{compatcond}
Combining~\er{intcov} with~\er{sys_intr}, we get 
\beq
\lb{forwxt}
\frac{\pd^2 w^j}{\pd x\pd t}=\sum_{l=1}^\dc\frac{\pd \al^j}{\pd w^l} w^l_t+
D_t(\al^j),\qquad\qquad
\frac{\pd^2 w^j}{\pd t\pd x}=\sum_{l=1}^\dc\frac{\pd \beta^j}{\pd w^l} w^l_x+
D_x(\beta^j),
\ee
where 
\beq
\lb{evdxdt}
D_x=\frac{\pd}{\pd x}+\sum_{\substack{i=1,\dots,\nv,\\ k\ge 0}} u^i_{k+1}\frac{\pd}{\pd u^i_k},\qquad\qquad
D_t=\frac{\pd}{\pd t}+\sum_{\substack{i=1,\dots,\nv,\\ k\ge 0}} D_x^k(F^i)\frac{\pd}{\pd u^i_k}
\ee
are the total derivative operators corresponding to~\er{sys_intr}. 
Using~\er{intcov} and \er{forwxt}, one obtains that the identity 
$\dfrac{\pd^2 w^j}{\pd x\pd t}=\dfrac{\pd^2 w^j}{\pd t\pd x}$ is equivalent to 
\beq
\lb{compat}
\sum_{l=1}^\dc\frac{\pd \al^j}{\pd w^l}\beta^l+D_t(\al^j)=
\sum_{l=1}^\dc\frac{\pd \beta^j}{\pd w^l}\al^l+D_x(\beta^j),\quad\qquad j=1,\dots,\dc.
\ee
System~\er{intcov} is called \emph{compatible modulo}~\eqref{sys_intr} if equations~\er{compat} 
hold for all values of the variables $x$, $t$, $u^i_k$, $w^l$.

Let $\ocn\in\zp$ be such that 
\beq
\notag
\frac{\pd \al^j}{\pd u^i_s}=0\qquad\forall\,s>\ocn,\qquad\forall\,i=1,\dots,\nv,\qquad\forall\,j=1,\dots,\dc.
\ee
Then equations~\er{compat} imply 
\beq
\notag
\frac{\pd \beta^j}{\pd u^i_r}=0\qquad\forall\,r>\ocn+\eo-1,\qquad\forall\,i=1,\dots,\nv,\qquad\forall\,j=1,\dots,\dc.
\ee

\end{remark}

\begin{remark}
In the case when $\nv=1$ and the functions $F^i$, $\al^j$, $\beta^j$ do not depend on $x$, $t$, 
the problem to describe compatible systems of the form 
\begin{gather}
\lb{covscal}
\begin{aligned}
w^j_x&=
\al^j(w^1,\dots,w^\dc,u,u_x,u_{xx},\dots),\\
w^j_t&=
\beta^j(w^1,\dots,w^\dc,u,u_x,u_{xx},\dots),
\end{aligned}\\
\notag
w^j=w^j(x,t),\qquad\qquad j=1,\dots,\dc, \qquad\qquad u=u^1.
\end{gather}
was studied in~\cite{cfa}.

In the case when~\er{sys_intr} is either the Burgers or the KdV equation, 
the problem to describe compatible systems of the form~\er{covscal} was also studied in~\cite{finley93}. 
However, gauge transformations were not considered in~\cite{finley93}. 
Because of this, the paper~\cite{finley93} had
to impose some additional constraints on the functions $\al^j$, $\beta^j$ in~\er{covscal}. 

\end{remark}

\begin{remark}
As has been said above, any compatible system~\er{intcov} 
is locally gauge equivalent to the system arising from a vector field representation of the algebra $\fdn^{\ocn}(\CE,a)$. (See Section~\ref{csev} for details.)

If the functions $\al^j$, $\be^j$ in~\er{intcov} are linear with respect to $w^1,\dots,w^\dc$, 
then \er{intcov} corresponds to a zero-curvature representation (ZCR) for system~\er{sys_intr}. 
So linear representations of the algebras $\fdn^{\ocn}(\CE,a)$ classify ZCRs up 
to local gauge transformations. 
In the case of scalar evolution equations, 
relations between $\fdn^{\ocn}(\CE,a)$ and ZCRs are studied in~\cite{lie-scal}. 

Some other approaches to the action of gauge transformations on ZCRs are described  
in~\cite{marv93,marv97,marv2010,sak95,sak2004,sebest2008} 
and references therein. 
Let $\mg$ be a finite-dimensional matrix Lie algebra. 
For a given $\mg$-valued ZCR, the papers~\cite{marv93,marv97,sak95} define 
certain $\mg$-valued functions that transform by conjugation when the ZCR transforms by gauge. 
Applications of these functions to construction and classification of 
some types of ZCRs are presented   
in~\cite{marv93,marv97,marv2010,sak95,sak2004,sebest2008} 
and references therein. 

To our knowledge, 
the theory of~\cite{marv93,marv97,marv2010,sak95,sak2004,sebest2008} 
does not produce any infinite-dimensional Lie algebras responsible for ZCRs. 
So this theory does not contain the algebras $\fdn^\ocn(\CE,a)$.
\end{remark}

\subsection{Necessary conditions for existence of B\"acklund transformations}
\lb{subsecbt}

It is well known that B\"acklund transformations 
are one of the main tools of soliton theory 
(see, e.g.,~\cite{backl82} and references therein). 
In this subsection we briefly discuss 
some applications of the algebras $\fdn^{\ocn}(\CE,a)$ to B\"acklund transformations. 


Let $\fdn(\CE,a)$ be the inverse (projective) 
limit of the sequence~\er{intfdoc1}. 
Since $\fdn^{\ocn}(\CE,a)$ in~\er{intfdoc1} are Lie algebras, 
the space $\fdn(\CE,a)$ is a Lie algebra as well.

\begin{remark}
Recall that $\CE$ is the infinite prolongation of system~\eqref{sys_intr}.
In the preprint~\cite{cfg} the algebra~$\fdn(\CE,a)$ is called the 
\emph{fundamental Lie algebra of $\CE$ at the point $a\in\CE$}. 

Let $M$ be a connected finite-dimensional manifold and $b\in M$. 
Consider the fundamental group $\pi_1(M,b)$.
The Lie algebra $\fdn(\CE,a)$ is called fundamental, 
because it is analogous to such fundamental groups 
in the following sense. 

According to~\cite{cfg,nonl}, 
there is a notion of coverings of PDEs such that 
compatible systems~\eqref{intcov} are coverings of~\eqref{sys_intr}.  
It is shown in~\cite{cfg,nonl} that 
this notion is somewhat similar to the classical concept of coverings from topology. 
Recall that
the fundamental group $\pi_1(M,b)$ is responsible for topological coverings of~$M$. 
According to~\cite{cfg}, 
the fundamental Lie algebra~$\fdn(\CE,a)$ plays a similar role  
for coverings~\eqref{intcov} of~\eqref{sys_intr}.  

It is shown in~\cite{cfg} that the algebra $\fdn(\CE,a)$ has some 
coordinate-independent geometric meaning. 

\end{remark}


Since $\fdn(\CE,a)$ is the inverse limit of~\er{intfdoc1}, 
for each $k\in\zp$ we have the natural surjective homomorphism 
$\rho_{k}\colon\fdn(\CE,a)\to\fdn^{k}(\CE,a)$. 
We define a topology on $\fdn(\CE,a)$ as follows.
For each $k\in\zp$ and each $v\in\fdn^k(\CE,a)$,  
the preimage~$\rho_k^{-1}(v)\subset\fdn(\CE,a)$ is, by definition, an open subset 
of $\fdn(\CE,a)$.
Such subsets form a base of the topology on $\fdn(\CE,a)$. 

A Lie subalgebra $H\subset\fdn(\CE,a)$ is said to be \emph{tame} 
if there are $k\in\zp$ and a subalgebra $\mathfrak{h}\subset\fdn^{k}(\CE,a)$ 
such that $H=\rho_{k}^{-1}(\mathfrak{h})$. 
Note that the codimension 
of $H$ in $\fdn(\CE,a)$ is equal to the codimension of $\mathfrak{h}$ 
in $\fdn^{k}(\CE,a)$.

\begin{remark}
It is easily seen that a subalgebra $H\subset\fdn(\CE,a)$ is tame iff 
$H$ is open and closed in $\fdn(\CE,a)$ with respect to the topology on $\fdn(\CE,a)$.
\end{remark}

A proof of the following proposition is sketched in~\cite{cfg}.
\begin{proposition}[\cite{cfg}]
\label{faexbt} 
Let $\CE_1$ and $\CE_2$ be evolution PDEs. 
Suppose that $\CE_1$ and $\CE_2$ are connected by a B\"acklund transformation. 
Then for each $i=1,2$ 
there are a point $a_i\in\CE_i$ and a tame subalgebra $H_i\subset\fdn(\CE_i,a_i)$ 
such that 
\begin{itemize}
\item the subalgebra $H_i$ is of finite codimension in $\fdn(\CE_i,a_i)$,
\item $H_1$ is isomorphic to $H_2$, and this isomorphism is a homeomorphism 
with respect to the topology induced by the embedding $H_i\subset\fdn(\CE_i,a_i)$. 
\end{itemize} 
\end{proposition}
The preprint~\cite{cfg} contains also a more general result 
about PDEs that are not necessarily evolution. 

Proposition~\ref{faexbt} provides 
a necessary condition for two given evolution PDEs 
to be connected by a B\"acklund transformation (BT). 
Using Proposition~\ref{faexbt}, one can prove non-existence of BTs for some PDEs. 

For example, the following result is obtained in~\cite{mpiprep10} by means of this theory. 

For any $e_1,e_2,e_3\in\Com$, 
consider the Krichever-Novikov equation $\hkneq(e_1,e_2,e_3)$ given by~\er{kneqprop} 
and the algebraic curve 
\beq
\lb{ceee}
C(e_1,e_2,e_3)=\Big\{(z,y)\in\Com^2\,\,\Big|\,\,
y^2=(z-e_1)(z-e_2)(z-e_3)\Big\}.
\ee
\begin{proposition}[\cite{mpiprep10}]\label{knprop}
Let $e_1,e_2,e_3,e'_1,e'_2,e'_3\in\Com$ be such that 
$e_i\neq e_j$ and $e'_i\neq e'_j$ for all $i\neq j$. 

If the curve $C(e_1,e_2,e_3)$ is not birationally equivalent to 
the curve $C(e'_1,e'_2,e'_3)$, 
then the equation $\hkneq(e_1,e_2,e_3)$ is not connected with the equation $\hkneq(e'_1,e'_2,e'_3)$ by any BT. 

Also, if $e_1\neq e_2\neq e_3\neq e_1$, then $\hkneq(e_1,e_2,e_3)$ is not 
connected with the KdV equation by any BT. 
\end{proposition}
Similar results are obtained in~\cite{mpiprep10} for the Landau-Lifshitz and 
nonlinear Schr\"odinger equations as well.

BTs of Miura type (differential substitutions) for
the Krichever-Novikov equation $\hkneq(e_1,e_2,e_3)$ are studied in~\cite{sok2013,svin-sok83}. 
According to~\cite{sok2013,svin-sok83}, the equation $\hkneq(e_1,e_2,e_3)$
is connected with the KdV equation by a BT of Miura type 
iff $e_i=e_j$ for some $i\neq j$.

The preprints~\cite{mpiprep10,cfg} and 
Propositions~\ref{faexbt},~\ref{knprop} consider the most general class of BTs, 
which is much larger than the class of 
BTs of Miura type studied in~\cite{sok2013,svin-sok83}.

Let $\CE_1$, $\CE_2$ be PDEs.
We say that $\CE_1$ and $\CE_2$ are \emph{BT-equivalent} 
if $\CE_1$ and $\CE_2$ are connected by a BT. 
It is known that this is a natural equivalence relation on the set of PDEs.

It is also known that, 
if $\CE_1$ and $\CE_2$ are connected by a BT,  
then these PDEs have similar properties. 
Therefore, it makes sense to try to classify PDEs up to BT-equivalence.
Let us consider some examples.

The paper~\cite{sok2013} presents a list of all (up to point transformations) 
scalar evolution equations of the form 
\beq
\lb{utu3}
u_t=u_{xxx}+f(x,u,u_x,u_{xx})
\ee
satisfying certain integrability conditions related 
to generalized symmetries and conservation laws. 
This result was announced in earlier papers of S.~I.~Svinolupov and V.~V.~Sokolov  
(see~\cite{sok2013} for references), but the proof is published in~\cite{sok2013}. 

This list of integrable PDEs~\er{utu3} in~\cite{sok2013} consists of 
so-called $S$-integrable and $C$-integrable equations.
(See~\cite{sok2013} for the definitions of $S$-integrable and $C$-integrable PDEs  
in the context of the above-mentioned integrability conditions. 
For example, the KdV and Krichever-Novikov equations are $S$-integrable.)

The list of $S$-integrable equations of the form~\er{utu3} is relatively long, 
but it can be simplified if one considers classification up to BT-equivalence as follows.


Let $\CE_1$, $\CE_2$ be PDEs of the form~\er{utu3}. 
We say that $\CE_1$ and $\CE_2$ are \emph{isomorphic} if these PDEs  
are connected by an invertible point transformation.
More precisely, we consider analytic point transformations
that are invertible on a nonempty open subset of the space of the variables $x$, $t$, $u$. 

Recall that an invertible point transformation can be regarded as a BT. 
(Although the class of BTs is much larger than the class of point transformations.)
Therefore, if $\CE_1$ and $\CE_2$ are isomorphic, then these PDEs are BT-equivalent. 


Let $\CE$ be an equation of the form~\er{utu3}.
We say that $\CE$ is \emph{of nonsingular Krichever-Novikov type} if there are 
$e_1,e_2,e_3\in\Com$, $e_1\neq e_2\neq e_3\neq e_1$, 
such that $\CE$ is isomorphic to the equation $\hkneq(e_1,e_2,e_3)$.

The following result is presented in~\cite{sok2013} 
(see also~\cite{svin-sok83} for similar results). 
\begin{proposition}[\cite{sok2013}]
\label{propmiura}
Let $\CE$ be an $S$-integrable equation of the form~\er{utu3}.

If $\CE$ is of nonsingular Krichever-Novikov type, then there is $\knpr\in\Com$, 
$\knpr\notin\{0,1\}$, such that $\CE$ is isomorphic to the equation $\hkneq(0,1,\knpr)$.

If $\CE$ is not of nonsingular Krichever-Novikov type, then $\CE$  
is connected with the KdV equation by a BT of Miura type 
\textup{(}a differential substitution\textup{)}. 

For $\knpr\notin\{0,1\}$, the equation $\hkneq(0,1,\knpr)$ is not 
connected with the KdV equation by any BT of Miura type. 
\end{proposition}

\begin{remark}
\lb{remmiura}
Let $\CE$ be an $S$-integrable PDE of the form~\er{utu3}.
Proposition~\ref{propmiura} implies that $\CE$ is BT-equivalent to either the KdV equation 
or the equation $\hkneq(0,1,\knpr)$ for some $\knpr\notin\{0,1\}$. 
However, Proposition~\ref{propmiura} is not sufficient to conclude 
that some of these PDEs are not BT-equivalent, 
because Proposition~\ref{propmiura} studies only a very particular class of BTs 
(BTs of Miura type). 

To determine which equations are not BT-equivalent, 
we need to use Proposition~\ref{knprop}, which considers the most general class of BTs. 
\end{remark}

Let $e_1,e_2,e_3,e'_1,e'_2,e'_3\in\Com$ be such that $e_1\neq e_2\neq e_3\neq e_1$ and  
$e'_1\neq e'_2\neq e'_3\neq e'_1$.

The set $\{e_1,e_2,e_3\}$ is \emph{affine-equivalent} 
to the set $\{e'_1,e'_2,e'_3\}$ if there are $b_1,b_2\in\Com$, $b_1\neq 0$, so that 
$b_1e_i+b_2\in\{e'_1,e'_2,e'_3\}$ for all $i=1,2,3$. 
In other words, the map $g\cl\Com\to\Com$ given by $g(z)=b_1z+b_2$ satisfies 
$\{g(e_1),g(e_2),g(e_3)\}=\{e'_1,e'_2,e'_3\}$.
Here $\{g(e_1),g(e_2),g(e_3)\}$ and $\{e'_1,e'_2,e'_3\}$ are unordered sets.

The next lemma is easy to check. It follows also from the results of~\cite{sok2013}.
\begin{lemma}
\lb{lemeee}
Let $e_1,e_2,e_3,e'_1,e'_2,e'_3\in\Com$ be such that 
$e_i\neq e_j$ and $e'_i\neq e'_j$ for all $i\neq j$. 

If $\{e_1,e_2,e_3\}$ is affine-equivalent to $\{e'_1,e'_2,e'_3\}$, 
then the equation $\hkneq(e_1,e_2,e_3)$ is isomorphic to $\hkneq(e'_1,e'_2,e'_3)$, 
and this isomorphism is given by a point transformation of the form
$$
x\mapsto c_1 x,\qquad t\mapsto c_2 t,\qquad u\mapsto c_3u+c_4,\qquad
c_1,c_2,c_3,c_4\in\Com,\qquad c_1c_2c_3\neq 0.
$$
\end{lemma}

The next proposition follows from the well-known classification of elliptic curves 
(see, e.g.,~\cite{hartshorne}).
\begin{proposition}[\cite{hartshorne}]
\lb{propla}
Recall that, for any $e_1,e_2,e_3\in\Com$, the algebraic curve $C(e_1,e_2,e_3)$ 
is given by~\er{ceee}.
Let $\knpr_1,\knpr_2\in\Com$ be such that $\knpr_i\notin\{0,1\}$ for $i=1,2$.

The curves $C(0,1,\knpr_1)$ and $C(0,1,\knpr_2)$ are birationally equivalent iff one has 
\beq
\lb{jla12}
\frac{\big((\knpr_1)^2-\knpr_1+1\big)^3}{(\knpr_1)^2(\knpr_1-1)^2}=
\frac{\big((\knpr_2)^2-\knpr_2+1\big)^3}{(\knpr_2)^2(\knpr_2-1)^2}.
\ee
The numbers $\knpr_1$, $\knpr_2$ satisfy~\er{jla12} iff 
the set $\{0,1,\knpr_1\}$ is affine-equivalent to the set $\{0,1,\knpr_2\}$.
\end{proposition}

Combining Propositions~\ref{knprop},~\ref{propmiura},~\ref{propla}, 
Remark~\ref{remmiura}, and Lemma~\ref{lemeee}, 
one obtains the following classification of 
$S$-integrable PDEs of the form~\er{utu3} up to BT-equivalence.
\begin{theorem}
\lb{btsint}
Let $\CE$ be an $S$-integrable equation of the form~\er{utu3}.
Then $\CE$ is BT-equivalent to either the KdV equation 
or the Krichever-Novikov equation $\hkneq(0,1,\knpr)$ for some $\knpr\notin\{0,1\}$. 

For any $\knpr\notin\{0,1\}$, 
the equation $\hkneq(0,1,\knpr)$ is not BT-equivalent to the KdV equation. 

Let $\knpr_1,\knpr_2\in\Com$ be such that $\knpr_i\notin\{0,1\}$ for $i=1,2$.
The equations $\hkneq(0,1,\knpr_1)$ and $\hkneq(0,1,\knpr_2)$ are BT-equivalent iff 
$\knpr_1$, $\knpr_2$ satisfy~\er{jla12}. 
Moreover, if $\knpr_1$, $\knpr_2$ satisfy~\er{jla12} then $\hkneq(0,1,\knpr_1)$ is 
isomorphic to $\hkneq(0,1,\knpr_2)$.
\end{theorem}

\subsection{Conventions and notation}
\lb{subs-conv}

The following conventions and notation are used in the paper.
All manifolds, functions, vector fields, and maps of manifolds are supposed to be complex-analytic.
The symbols $\zsp$ and $\zp$ denote the sets of positive and nonnegative 
integers respectively.

\section{Coverings of $(1+1)$-dimensional evolution PDEs}
\label{csev}

\subsection{Coverings and gauge transformations}

Consider an evolution PDE 
\begin{gather}
\label{gev}
\frac{\pd u^i}{\pd t}
=F^i(x,t,u^1,\dots,u^\nv,\,u^1_1,\dots,u^\nv_1,\dots,u^1_{\eo},\dots,u^\nv_{\eo}),\\
\notag
u^i=u^i(x,t),\quad\qquad u^i_k=\frac{\pd^k u^i}{\pd x^k},\quad\qquad  
i=1,\dots,\nv. 
\end{gather}
Recall that the infinite prolongation $\CE$ of~\er{gev} is the  
infinite-dimensional manifold 
with the coordinates $x$, $t$, $u^i_k$ for $i=1,\dots,\nv$ and $k\in\zp$. Here $u^i_0=u^i$. 


In what follows, when we consider a function of the variables $u^i_k$, we always assume 
that the function may depend only on a finite number of these variables.
The total derivative operators $D_x$, $D_t$ given by formulas~\er{evdxdt} 
are viewed as vector fields on the manifold $\CE$. 

Suppose that a system 
\begin{gather}
\lb{naiv}
\begin{aligned}
w^j_x&=
\al^j(w^1,\dots,w^\dc,x,t,u^i_k,\dots),\\
w^j_t&=
\beta^j(w^1,\dots,w^\dc,x,t,u^i_k,\dots),
\end{aligned}\\
\notag
w^j=w^j(x,t),\qquad\qquad j=1,\dots,\dc, 
\end{gather}
is compatible modulo~\eqref{gev}.

Let $W$ be the manifold with coordinates $w^1,\dots,w^\dc$. 
Then the expressions 
\begin{gather}
\lb{vfaw}
A=\sum_{j=1}^\dc \al^j(w^1,\dots,w^\dc,x,t,u^i_k,\dots)\frac{\partial}{\pd w^j},\\
\lb{vfbw}
B=\sum_{j=1}^\dc \beta^j(w^1,\dots,w^\dc,x,t,u^i_k,\dots)\frac{\partial}{\pd w^j}
\end{gather}
can be regarded as vector fields on the manifold $\CE\times W$.

The compatibility condition~\er{compat} of system~\er{naiv} is equivalent to the equation 
\beq
\label{gc}
D_x(B)-D_t(A)+[A,B]=0,
\ee
where $D_x(B)=\sum_{j=1}^\dc D_x(\beta^j)\dfrac{\partial}{\pd w^j}$ and 
$D_t(A)=\sum_{j=1}^\dc D_t(\al^j)\dfrac{\partial}{\pd w^j}$.

If system~\er{naiv} is compatible modulo~\eqref{gev}, then~\er{naiv} is called a \emph{covering} of~\eqref{gev}. 
Covering~\er{naiv} is uniquely determined by the vector fields $A$, $B$ given by formulas~\er{vfaw}, \er{vfbw}.

\begin{remark}
This definition of coverings is a particular case of a more general concept of coverings of PDEs from~\cite{nonl}.  
\end{remark}


A covering~\er{naiv} is said to be \emph{of order not greater than $\ocn\in\zp$} if
the functions $\al^j$ may depend only on the variables $w^l$, $x$, $t$, $u^i_{k}$ for $k\le\ocn$.
In other words, covering~\er{naiv} is of order $\le\ocn$ iff the vector field~\er{vfaw} satisfies 
\beq
\lb{orderna}
\frac{\pd A}{\pd u^i_s}=0\qquad\forall\,s>\ocn,\qquad\forall\,i=1,\dots,\nv.
\ee
If~\er{orderna} holds, then equation~\er{gc} implies 
\beq
\lb{ordernb}
\frac{\pd B}{\pd u^i_r}=0
\qquad\quad\forall\,r>\ocn+\eo-1,\qquad\forall\,i=1,\dots,\nv.
\ee


As has already been said in Section~\ref{main_int},
a gauge transformation is given by an invertible change of variables
\begin{equation}
\label{gtgg}
  x\mapsto x,\qquad t\mapsto t,\qquad u^i_k\mapsto u^i_k,\qquad
  w^j\mapsto g^j(\tilde w^1,\dots,\tilde w^\dc,x,t,u^i_l,\dots),\qquad j=1,\dots,\dc.
\end{equation}

Substituting~\er{gtgg} to~\er{naiv}, we obtain a system of the form
\begin{gather}
\lb{tcovmain}
\begin{aligned}
\tilde w^j_x&=
\tilde\al^j(\tilde w^1,\dots,\tilde w^\dc,x,t,u^i_k,\dots),\\
\tilde w^j_t&=
\tilde\beta^j(\tilde w^1,\dots,\tilde w^\dc,x,t,u^i_{k},\dots),
\end{aligned}\\
\notag
\tilde w^j=\tilde w^j(x,t),\qquad\qquad j=1,\dots,\dc. 
\end{gather}  
Covering~\er{tcovmain} is said to be \emph{gauge equivalent} to covering~\er{naiv} if 
\er{tcovmain} and \er{naiv} are connected by a gauge transformation~\er{gtgg}.

\begin{example}
\label{ex_gt}
Consider a scalar evolution equation  
\beq
\label{utf1}
u_t=F(x,t,u,u_1,\dots,u_\eo),\qquad\quad u=u(x,t),\qquad\quad
u_k=\frac{\pd^k u}{\pd x^k}. 
\ee
Then $D_x$ is given by the formula 
$D_x=\dfrac{\pd}{\pd x}+\sum_{k\ge 0} u_{k+1}\dfrac{\pd}{\pd u_k}$, where $u_0=u$.

Let $\dc=1$ and $w=w^1$. Consider a covering 
\begin{equation}
\label{ex_naiv}
w_x=\al(w,x,t,u_0,u_1,\dots),\qquad\qquad
w_t=\beta(w,x,t,u_0,u_1,\dots).
\end{equation}
We want to determine how covering~\er{ex_naiv} 
changes after a gauge transformation of the form 
\begin{equation}
\label{gt1}
x\mapsto x,\qquad t\mapsto t,\qquad u_k\mapsto u_k,\qquad 
w\mapsto g(\tilde w,x,t,u_0,u_1),\qquad
\frac{\pd g}{\pd \tilde w}\neq 0.
\end{equation}
We need to substitute $g(\tilde w,x,t,u_0,u_1)$ in place of~$w$ in equations~\er{ex_naiv}. 
The result is  
\begin{gather}
\label{covg1}
\begin{aligned}
\frac{\pd g}{\pd \tilde w}\cdot \tilde w_x+\frac{\pd g}{\pd x}
+\frac{\pd g}{\pd u_0}\cdot u_1+\frac{\pd g}{\pd u_1}\cdot u_2&=\al(g,x,t,u_0,u_1,\dots),\\
\frac{\pd g}{\pd \tilde w}\cdot \tilde w_t+\frac{\pd g}{\pd t}
+\frac{\pd g}{\pd u_0}\cdot u_t+\frac{\pd g}{\pd u_1}\cdot u_{xt}&=\beta(g,x,t,u_0,u_1,\dots),
\end{aligned}\\
\notag
g=g(\tilde w,x,t,u_0,u_1).
\end{gather}

Since $u_t=F$ and $u_{xt}=D_x(F)$ due to equation~\er{utf1}, 
system~\er{covg1} can be written as 
\begin{gather}
\label{covtr1}
\begin{aligned}
\tilde w_x&=
\frac{1}{g_{\tilde w}}\Bigl(\al(g,x,t,u_0,u_1,\dots)-\frac{\pd g}{\pd x}
-u_1\frac{\pd g}{\pd u_0}-u_2\frac{\pd g}{\pd u_1}\Bigl),\\
\tilde w_t&=
\frac{1}{g_{\tilde w}}\Bigl(\beta(g,x,t,u_0,u_1,\dots)-\frac{\pd g}{\pd t}
-F\frac{\pd g}{\pd u_0}-D_x(F)\frac{\pd g}{\pd u_1}\Bigl),
\end{aligned}\\
\notag
g=g(\tilde w,x,t,u_0,u_1),\qquad\qquad g_{\tilde w}=\frac{\pd g}{\pd \tilde w}.
\end{gather}
Thus, applying the gauge transformation~\er{gt1} to covering~\er{ex_naiv}, 
one obtains covering~\er{covtr1}.
\end{example}

Return to the general case of system~\er{gev} and covering~\er{naiv}.
Suppose that \er{tcovmain} is obtained from~\er{naiv} by means of a gauge transformation~\er{gtgg}. 
Let us present explicit formulas for the functions $\tilde\al^j$, $\tilde\beta^j$ from \er{tcovmain}.

In order to apply the gauge transformation~\er{gtgg} to covering~\er{naiv}, 
we need to substitute $g^j(\tilde w^1,\dots,\tilde w^\dc,x,t,u^i_l,\dots)$ in place of~$w^j$ in equations~\er{naiv}. 
The result is  
\begin{gather*}
\sum_{r=1}^{\dc}\frac{\pd g^j}{\pd \tilde w^r}\cdot \tilde w^r_x+D_x(g^j)
=\al^j(g^1,\dots,g^\dc,x,t,u^i_k,\dots),\\
\sum_{r=1}^{\dc}\frac{\pd g^j}{\pd \tilde w^r}\cdot \tilde w^r_t+D_t(g^j)
=\beta^j(g^1,\dots,g^\dc,x,t,u^i_k,\dots),\\
g^j=g^j(\tilde w^1,\dots,\tilde w^\dc,x,t,u^i_l,\dots),\qquad\qquad j=1,\dots,\dc.
\end{gather*}
Note that the matrix 
$
\begin{pmatrix}
\frac{\pd g^1}{\pd \tilde w^1} & \ldots & \frac{\pd g^1}{\pd \tilde w^\dc}\\
\vdots & \ddots & \vdots \\
\frac{\pd g^\dc}{\pd \tilde w^1} & \ldots & \frac{\pd g^\dc}{\pd \tilde w^\dc}
\end{pmatrix}
$
is invertible, because the transformation~\er{gtgg} is supposed to be invertible.
Therefore, applying the gauge transformation~\er{gtgg} to system~\er{naiv}, 
we obtain the system
\begin{gather*}
\label{wxmatr}
\begin{pmatrix}
\tilde w^1_x\\
\vdots\\
\tilde w^\dc_x 
\end{pmatrix}=
\begin{pmatrix}
\frac{\pd g^1}{\pd \tilde w^1} & \ldots & \frac{\pd g^1}{\pd \tilde w^\dc}\\
\vdots & \ddots & \vdots \\
\frac{\pd g^\dc}{\pd \tilde w^1} & \ldots & \frac{\pd g^\dc}{\pd \tilde w^\dc}
\end{pmatrix}^{-1}\cdot
\begin{pmatrix}
\al^1(g^1,\dots,g^\dc,x,t,u^i_k,\dots)-D_x(g^1)\\
\vdots\\ 
\al^\dc(g^1,\dots,g^\dc,x,t,u^i_k,\dots)-D_x(g^\dc)
\end{pmatrix},\\
\label{wtmatr}
\begin{pmatrix}
\tilde w^1_t\\
\vdots\\
\tilde w^\dc_t 
\end{pmatrix}=
\begin{pmatrix}
\frac{\pd g^1}{\pd \tilde w^1} & \ldots & \frac{\pd g^1}{\pd \tilde w^\dc}\\
\vdots & \ddots & \vdots \\
\frac{\pd g^\dc}{\pd \tilde w^1} & \ldots & \frac{\pd g^\dc}{\pd \tilde w^\dc}
\end{pmatrix}^{-1}\cdot
\begin{pmatrix}
\beta^1(g^1,\dots,g^\dc,x,t,u^i_k,\dots)-D_t(g^1)\\
\vdots\\ 
\beta^\dc(g^1,\dots,g^\dc,x,t,u^i_k,\dots)-D_t(g^\dc)
\end{pmatrix}.
\end{gather*}
Hence the functions $\tilde\al^j$, $\tilde\beta^j$ from \er{tcovmain} are given by the formulas 
\begin{gather}
\lb{alform}
\begin{pmatrix}
{\tilde{\al}}^1\\
\vdots\\
{\tilde{\al}}^\dc 
\end{pmatrix}=
\begin{pmatrix}
\frac{\pd g^1}{\pd \tilde w^1} & \ldots & \frac{\pd g^1}{\pd \tilde w^\dc}\\
\vdots & \ddots & \vdots \\
\frac{\pd g^\dc}{\pd \tilde w^1} & \ldots & \frac{\pd g^\dc}{\pd \tilde w^\dc}
\end{pmatrix}^{-1}\cdot
\begin{pmatrix}
\al^1(g^1,\dots,g^\dc,x,t,u^i_k,\dots)-D_x(g^1)\\
\vdots\\ 
\al^\dc(g^1,\dots,g^\dc,x,t,u^i_k,\dots)-D_x(g^\dc)
\end{pmatrix},\\
\lb{betform}
\begin{pmatrix}
{\tilde{\beta}}^1\\
\vdots\\
{\tilde{\beta}}^\dc 
\end{pmatrix}=
\begin{pmatrix}
\frac{\pd g^1}{\pd \tilde w^1} & \ldots & \frac{\pd g^1}{\pd \tilde w^\dc}\\
\vdots & \ddots & \vdots \\
\frac{\pd g^\dc}{\pd \tilde w^1} & \ldots & \frac{\pd g^\dc}{\pd \tilde w^\dc}
\end{pmatrix}^{-1}\cdot
\begin{pmatrix}
\beta^1(g^1,\dots,g^\dc,x,t,u^i_k,\dots)-D_t(g^1)\\
\vdots\\ 
\beta^\dc(g^1,\dots,g^\dc,x,t,u^i_k,\dots)-D_t(g^\dc)
\end{pmatrix},\\
\notag
g^j=g^j(\tilde w^1,\dots,\tilde w^\dc,x,t,u^i_l,\dots),\qquad\qquad j=1,\dots,\dc.
\end{gather}

Let $\tilde W$ be the manifold with coordinates $\tilde w^1,\dots,\tilde w^\dc$. 
Formulas~\er{gtgg} determine the diffeomorphism
\begin{gather}
\lb{gdef1}
G\cl\CE\times \tilde W\to\CE\times W,\qquad\qquad 
G^*(x)=x,\qquad G^*(t)=t,\qquad G^*(u^i_k)=u^i_k,\\
\lb{gdef2}
G^*(w^j)=g^j(\tilde w^1,\dots,\tilde w^\dc,x,t,u^i_l,\dots),
\end{gather}
where $G^*$ is the pull-back map corresponding to the diffeomorphism $G$.

According to~\er{alform},~\er{betform},~\er{gdef1},~\er{gdef2}, for the vector fields 
$\tilde A=\sum_{j=1}^\dc {\tilde{\al}}^j\dfrac{\partial}{\pd \tilde w^j}$ and 
$\tilde B=\sum_{j=1}^\dc {\tilde{\beta}}^j\dfrac{\partial}{\pd \tilde w^j}$ we have 
$$
G_*\big(D_x+\tilde A\big)=D_x+A,\qquad\qquad 
G_*\big(D_t+\tilde B\big)=D_t+B,
$$
where $G_*$ is the differential of the diffeomorphism~$G$, 
and the vector fields $A$, $B$ are given by \er{vfaw}, \er{vfbw}.

To simplify notation, we identify $\tilde w^j$ with $w^j$. 
A gauge transformation given by~\er{gtgg} will be written simply as 
$$
w^j\mapsto g^j(w^1,\dots,w^\dc,x,t,u^i_l,\dots),\qquad\qquad 
j=1,\dots,\dc.
$$




\subsection{Normal forms of coverings with respect to the action of gauge transformations}
\lb{sec_normform}

Recall that $W$ is the manifold with coordinates $w^1,\dots,w^\dc$. 

It is convenient to say that a covering is given by vector fields $D_x+A$, $D_t+B$ on the manifold $\CE\times W$, 
where $A$, $B$ are of the form~\er{vfaw},~\er{vfbw} for some functions $\al^j$, $\beta^j$ and satisfy~\er{gc}. Note that equation~\er{gc} is equivalent to 
$[D_x+A,\,D_t+B]=0$.
 
Recall that a covering is of order $\le\ocn$ iff $A$, $B$ 
satisfy~\er{orderna},~\er{ordernb}.

We want to find a normal form for coverings 
with respect to the action of the group of gauge transformations.
Consider first the case $\nv=1$, and set $u=u^1$.
Then the coordinates on $\CE$ are $x$, $t$, $u_k$, $k\in\zp$.

A point $a\in\CE$ is determined by the values of the coordinates $x$, $t$, $u_k$ at $a$. 
Let
\begin{equation}
\notag
a=(x=x_0,\,t=t_0,\,u_k=a_k)\,\in\,\CE,\qquad x_0,\,t_0,\,a_k\in\Com,\qquad k\in\zp,
\end{equation}
be a point of $\CE$.


\begin{remark}
\lb{subs_not}
Let $F$ be a function of the variables $x$, $t$, $u_k$. 
Let $s\in\zp$. 
Then the notation 
$$
F\,\Big|_{u_k=a_k,\ k\ge s}
$$
means that we substitute $u_k=a_k$ for all $k\ge s$ in the function $F$. 

Also, sometimes we need to substitute $x=x_0$ or $t=t_0$. 
For example, if $F=F(x,t,u_0,u_1,u_2,u_3)$, then 
$$
F\,\Big|_{x=x_0,\ u_k=a_k,\ k\ge 2}=F(x_0,t,u_0,u_1,a_2,a_3).
$$
\end{remark}

\begin{theorem}
\lb{d1evcov}
Fix a covering of order $\le\ocn$. 
For any $b\in W$, on a neighborhood of $(a,b)\in\CE\times W$
there is a unique gauge transformation
\beq
\lb{wimfi}
w^j\mapsto g^j(w^1,\dots,w^\dc,x,t,u_0,u_1,\dots),
\qquad\qquad j=1,\dots,\dc,
\ee
such that
\begin{itemize}
\item the transformed vector fields $D_x+A$, $D_t+B$ satisfy for all $s\ge 1$
\begin{gather}
\label{d=0}
\frac{\pd A}{\pd u_s}\,\,\bigg|_{u_k=a_k,\ k\ge s}=0,\\
\lb{aukak}
A\,\Big|_{u_k=a_k,\ k\ge 0}=0,\\
\lb{bxx0}
B\,\Big|_{x=x_0,\ u_k=a_k,\ k\ge 0}=0,
\end{gather}
\item one has
\beq
\lb{f=wi}
g^j\,\Big|_{x=x_0,\ t=t_0,\ u_k=a_k,\ k\ge 0}=w^j,\qquad j=1,\dots,\dc.
\ee
\end{itemize}

Moreover, this gauge transformation obeys
\beq
\lb{pdgjuk}
\frac{\pd g^j}{\pd u_k}=0\quad\qquad\forall\,k\ge \ocn,\quad\qquad j=1,\dots,\dc,
\ee
and the transformed covering is also of order $\le p$.
\end{theorem}
\begin{proof}  
Suppose that the initial covering is given by vector fields $D_x+A$, $D_t+B$, 
where $A$, $B$ do not necessarily satisfy~\er{d=0}, \er{aukak}, \er{bxx0}. 

We are going to construct a gauge transformation of the form~\er{wimfi},~\er{f=wi},~\er{pdgjuk}
such that the transformed vector fields $D_x+A$, $D_t+B$ will satisfy~\er{aukak},~\er{bxx0}, 
and~\er{d=0} for all $s\ge 1$.

We are going to construct the required gauge transformation in several steps.
First, we will construct a transformation to achieve property~\er{d=0},
then another transformation to get properties~\er{d=0},~\er{aukak},
and finally another transformation to obtain all properties~\er{d=0},~\er{aukak},~\er{bxx0}.

Let us first prove that  
after a suitable gauge transformation
one gets~\er{d=0} for all $s\ge 1$. 

Since the covering is of order $\le\ocn$, equation~\er{d=0} is valid for all $s>\ocn$. 
Let $n\in\{1,\dots,\ocn\}$ be such that~\er{d=0} holds for all $s\ge n+1$. 
It is easily seen that this property is preserved by 
any gauge transformation of the form 
\begin{equation}
\label{gtn}
w^j\mapsto \tilde g^j(w^1,\dots,w^\dc,x,t,u_0,\dots,u_{n-1}),\qquad\qquad j=1,\dots,\dc.
\end{equation}
Therefore, if we find a gauge transformation~\er{gtn} 
such that after this transformation we get~\er{d=0} for $s=n$, 
then we will get~\er{d=0} for all $s\ge n$.
 
One has 
\begin{equation*}
\frac{\pd A}{\pd u_n}\,\,\bigg|_{u_k=a_k,\ k\ge n}
=\sum_{j=1}^\dc c^j(w^1,\dots,w^\dc,x,t,u_0,\dots,u_{n-1})\frac{\partial}{\pd w^j} 
\end{equation*}
for some functions $c^j(w^1,\dots,w^\dc,x,t,u_0,\dots,u_{n-1})$.
Consider the system of ordinary differential equations (ODE) 
with respect to the variable $u_{n-1}$  
\begin{gather*}
\frac{\pd}{\pd u_{n-1}}\tilde g^j(w^1,\dots,w^\dc,x,t,u_0,\dots,u_{n-1})=
c^j(\tilde g^1,\dots,\tilde g^\dc,x,t,u_0,\dots,u_{n-1}),\\
j=1,\dots,\dc, 
\end{gather*} 
for unknown functions $\tilde g^j$. 
Here $w^1,\dots,w^\dc$, $x$, $t$, $u_0,\dots,u_{n-2}$ are regarded as parameters. 
A local solution of this ODE with the initial condition 
$$
\tilde g^j(w^1,\dots,w^\dc,x,t,u_0,\dots,u_{n-2},a_{n-1})=w^j,\quad\qquad j=1,\dots,\dc,
$$
determines transformation~\er{gtn} such that after this transformation we get~\er{d=0} for all $s\ge n$. 
Using induction, we obtain that, after a suitable gauge transformation, property~\er{d=0} is valid for all $s\ge 1$.

Clearly, property~\er{d=0} is preserved by any gauge transformation of the form 
\begin{equation}
\label{gtxt}
w^j\mapsto \hat g^j(w^1,\dots,w^\dc,x,t),\qquad\qquad j=1,\dots,\dc. 
\end{equation} 
Let us find a gauge transformation of the form~\er{gtxt}  
such that after this transformation we get~\er{aukak}. 
We have 
\begin{equation*}
A\,\Big|_{u_k=a_k,\ k\ge 0}
=\sum_{j=1}^\dc h^j(w^1,\dots,w^\dc,x,t)\frac{\partial}{\pd w^j} 
\end{equation*} 
for some functions $h^j(w^1,\dots,w^\dc,x,t)$. 
Consider the ODE with respect to the variable $x$ 
$$
\frac{\pd}{\pd x}\hat g^j(w^1,\dots,w^\dc,x,t)=h^j(\hat g^1,\dots,\hat g^\dc,x,t),
\qquad j=1,\dots,\dc, 
$$
where $w^1,\dots,w^\dc,t$ are treated as parameters.  
Its local solution with the initial condition 
$$
\hat g^j(w^1,\dots,w^\dc,x_0,t)=w^j,\qquad\quad j=1,\dots,\dc,
$$
determines the required transformation~\er{gtxt}.

Properties~\er{d=0},~\er{aukak} are preserved by any gauge transformation of the form 
\begin{equation}
\label{gtt}
w^j\mapsto \check{g}^j(w^1,\dots,w^\dc,t),\qquad\quad j=1,\dots,\dc. 
\end{equation} 
One has 
\begin{equation*}
B\,\Big|_{x=x_0,\ u_k=a_k,\ k\ge 0}
=\sum_{j=1}^\dc f^j(w^1,\dots,w^\dc,t)\frac{\partial}{\pd w^j} 
\end{equation*} 
for some functions $f^j(w^1,\dots,w^\dc,t)$. 
Consider the ODE with respect to $t$
$$
\frac{\pd}{\pd t}\check{g}^j(w^1,\dots,w^\dc,t)=f^j(\check{g}^1,\dots,\check{g}^\dc,t),
\qquad\quad j=1,\dots,\dc,
$$
where $w^1,\dots,w^\dc$ are viewed as parameters.
Its local solution with the initial condition 
$$
\check{g}^j(w^1,\dots,w^\dc,t_0)=w^j,\qquad\quad j=1,\dots,\dc,
$$
determines a gauge transformation of the form~\er{gtt} 
such that the transformed vector field $D_x+B$ satisfies~\er{bxx0}. 

Thus we have found a gauge transformation of the form~\er{wimfi},~\er{f=wi},~\er{pdgjuk} 
such that the transformed vector fields $D_x+A,\ D_t+B$ obey~\er{aukak},~\er{bxx0}, 
and~\er{d=0} for all $s\ge 1$. 
Since we have applied this transformation to a covering of order $\le p$, 
equation~\er{pdgjuk} implies that the transformed covering is also of order $\le p$. 

It remains to prove uniqueness of such a gauge transformation. 

Consider a covering given by $D_x+A$, $D_t+B$ such that
$A$, $B$ satisfy~\er{aukak},~\er{bxx0}, and~\er{d=0} for all $s\ge 1$. 
Consider a gauge transformation of the form 
\begin{gather*}
w^j\mapsto\bar g^j(w^1,\dots,w^\dc,x,t,u_0,u_1,\dots),\\
\bar g^j\,\Big|_{x=x_0,\ t=t_0,\ u_k=a_k,\ k\ge 0}=w^j,\quad\qquad j=1,\dots,\dc,
\end{gather*}
such that, applying this transformation to $D_x+A$, $D_t+B$, we
get vector fields $D_x+A'$, $D_t+B'$, where $A'$, $B'$ 
obey properties~\er{d=0},~\er{aukak},~\er{bxx0} as well. 

We need to show that 
\begin{gather}
\lb{jbargj}
\forall\,j\qquad\frac{\pd\bar g^j}{\pd u_k}=0\qquad\forall\,k\in\zp,\\ 
\lb{bargjx}
\forall\,j\qquad\frac{\pd\bar g^j}{\pd x}=0,\\
\lb{bargjt}
\forall\,j\qquad\frac{\pd\bar g^j}{\pd t}=0.
\end{gather} 
Suppose that~\er{jbargj} does not hold. Let $l$ be the maximal integer such that 
$\dfrac{\pd\bar g^j}{\pd u_l}\neq 0$ for some $j$. Then it is easily seen that $A'$ 
does not satisfy~\er{d=0} for $s=l+1$. 

If~\er{jbargj} is valid and~\er{bargjx} is not, then $A'$ does not obey~\er{aukak}. 
Finally, if~\er{jbargj},~\er{bargjx} hold and \er{bargjt} does not, then $B'$ does not satisfy~\er{bxx0}.
\end{proof}

Return to the case of arbitrary $\nv$ and the coordinate system $x$, $t$, $u^i_k$ for $\CE$.
Let
\begin{equation}
\label{pointev}
a=(x=x_0,\,t=t_0,\,u^i_k=a^i_k)\,\in\,\CE,\qquad x_0,\,t_0,\,a^i_k\in\Com,\qquad
i=1,\dots,\nv,\qquad k\in\zp,
\end{equation}
be a point of $\CE$.
We want to obtain an analog of Theorem~\ref{d1evcov} for arbitrary $\nv$.


Consider the following ordering $\preceq$ of the set 
$\{1,\dots,\nv\}\times\zp$ 
\begin{gather}
\notag
i,i'\in\{1,\dots,\nv\},\qquad\quad k,k'\in\zp,\qquad\quad k\neq k',\\
\lb{ev_ord}
(i,k)\prec(i',k')\ \text{ iff }\ k<k',\qquad\qquad 
(i,k)\prec(i',k)\ \text{ iff }\ i<i'.
\end{gather}
That is, $(1,0)\prec(2,0)\prec\dots\prec(\nv,0)\prec(1,1)\prec(2,1)\prec\dots$.

As usual, the notation $(i_1,k_1)\succeq(i_2,k_2)$ means 
that either $(i_1,k_1)\succ(i_2,k_2)$ or $(i_1,k_1)=(i_2,k_2)$. 

\begin{remark}
Let $F$ be a function of the variables $x$, $t$, $u^i_k$. 
Let $i'\in\{1,\dots,\nv\}$ and $k'\in\zp$. 
Then the notation 
$$
F\,\Big|_{u^i_k=a^i_k\ \forall\,(i,k)\succ(i',k')}
$$
says that 
we substitute $u^i_k=a^i_k$ for all $(i,k)\succ(i',k')$ in the function $F$. 

Similarly, the notation 
$$
F\,\Big|_{x=x_0,\ u^i_k=a^i_k\ \forall\,(i,k)\succeq(i',k')}
$$
means that we substitute $x=x_0$ and $u^i_k=a^i_k$ for all $(i,k)\succeq(i',k')$ in $F$. 
\end{remark}

\begin{theorem}
\lb{evcov}
Fix a covering of order $\le\ocn$.
For any $b\in W$, on a neighborhood of $(a,b)\in\CE\times W$ there is a unique gauge transformation 
\beq
\lb{gwimfi}
w^j\mapsto g^j(w^1,\dots,w^\dc,x,t,u^i_k,\dots),
\qquad\qquad j=1,\dots,\dc,
\ee
such that
\begin{itemize}
\item the transformed vector fields $D_x+A$, $D_t+B$ satisfy
for all $i_0=1,\dots,\nv$ and $k_0\in\zsp$
\begin{gather}
\label{gd=0}
\frac{\pd  A}{\pd u^{i_0}_{k_0}}\,\,\bigg|_{u^i_k=a^i_k\ \forall\,(i,k)\succ(i_0,k_0-1)}=0,\\
\lb{gaukak}
 A\,\Big|_{u^i_k=a^i_k\ \forall\,(i,k)\succeq(1,0)}=0,\\
\lb{gbxx0}
{B}\,\Big|_{x=x_0,\ u^i_k=a^i_k\ \forall\,(i,k)\succeq(1,0)}=0,
\end{gather}
\item one has
\beq
\lb{gf=wi}
g^j\,\Big|_{x=x_0,\ t=t_0,\ u^i_k=a^i_k\ \forall\,(i,k)\succeq(1,0)}=w^j,\qquad j=1,\dots,\dc.
\ee
\end{itemize}

Moreover, this gauge transformation obeys
\beq
\lb{gpdgjuk}
\frac{\pd g^j}{\pd u^i_k}=0\qquad\forall\,k\ge \ocn,\qquad i=1,\dots,\nv,\qquad j=1,\dots,\dc,
\ee
and the transformed covering is also of order $\le p$.

\end{theorem}
\begin{proof} 
Note that for $\nv=1$ this theorem is equivalent to Theorem~\ref{d1evcov}. 

Suppose that the initial covering is given by vector fields $D_x+A$, $D_t+B$, 
where $A$, $B$ do not necessarily satisfy~\er{gd=0}, \er{gaukak}, \er{gbxx0}. 
Since the covering is of order $\le\ocn$, we have~\er{orderna}.

Similarly to the proof of Theorem~\ref{d1evcov}, 
we are going to construct a gauge transformation of the form~\er{gwimfi},~\er{gf=wi},~\er{gpdgjuk}
such that the transformed vector fields $D_x+A$, $D_t+B$ will satisfy~\er{gaukak},~\er{gbxx0}, 
and~\er{gd=0} for all $i_0=1,\dots,\nv$ and $k_0\in\zsp$.

Let us first prove that after a suitable gauge transformation 
one gets property~\er{gd=0} for all $i_0=1,\dots,\nv$ and $k_0\in\zsp$. 

Let $(i',k')$ be the minimal element with respect to the ordering~\er{ev_ord} 
such that property~\er{gd=0} holds for all $(i_0,k_0)\succ(i',k')$. 
The minimal element exists, because $A$ obeys~\er{orderna}.

If $k'=0$, then~\er{gd=0} is valid for all $i_0=1,\dots,\nv$ and $k_0\in\zsp$.


Consider the case~$k'>0$. 
We have 
\beq
\label{dauik'}
\frac{\pd A}{\pd u^{i'}_{k'}}\,\,\bigg|_{u^i_k=a^i_k\ \forall\,(i,k)\succ(i',k'-1)}
=\sum_{j=1}^\dc c^j(w^1,\dots,w^\dc,x,t,u^{i_1}_{k_1},\dots)\frac{\pd}{\pd w^j} 
\ee 
for some functions $c^j$, which may depend on the following variables 
\beq
\lb{wxtui1}
w^1,\dots,w^\dc,\quad x,\quad t,\quad u^{i_1}_{k_1},\quad (i_1,k_1)\preceq(i',k'-1).
\ee

Let us find a gauge transformation of the form 
\beq
\lb{gtgjui1}
w^j\mapsto {\tilde{g}}^j(w^1,\dots,w^\dc,x,t,u^{i_1}_{k_1},\dots),
\quad\qquad j=1,\dots,\dc,
\ee
such that after this transformation we get property~\er{gd=0} for all  $(i_0,k_0)\succeq(i',k')$. 
We assume that functions ${\tilde{g}}^j$ in~\er{gtgjui1} 
may depend only on the variables~\er{wxtui1}.

It is easy to check that such a transformation must satisfy the equations
\beq
\label{dtguk1}
\frac{\pd}{\pd u^{i'}_{k'-1}}{\tilde{g}}^j(w^1,\dots,w^\dc,x,t,u^{i_1}_{k_1},\dots)=
c^j({\tilde{g}}^1,\dots,{\tilde{g}}^\dc,x,t,u^{i_1}_{k_1},\dots),
\qquad j=1,\dots,\dc. 
\ee
We regard~\er{dtguk1} as a parameter-dependent system of ordinary differential equations (ODE)
with respect to the variable~$u^{i'}_{k'-1}$ and unknown functions~${\tilde{g}}^j$, 
where $w^1,\dots,w^\dc$, $x$, $t$, $u^{i_2}_{k_2}$ for $(i_2,k_2)\prec(i',k'-1)$ are viewed as parameters. 

Since we are interested in gauge transformations satisfying~\er{gf=wi}, 
we choose the following initial condition for this ODE
\beq
\label{gcincond}
{\tilde{g}}^j\,\Big|_{u^{i'}_{k'-1}=a^{i'}_{k'-1}}=w^j,\qquad\qquad j=1,\dots,\dc.
\ee
Then ${\tilde{g}}^1,\dots,{\tilde{g}}^\dc$ are defined as a solution of the ODE~\er{dtguk1} with the initial condition~\er{gcincond}.

By induction with respect to the ordering~\er{ev_ord}, 
we obtain that, after a suitable gauge transformation, 
property~\er{gd=0} is valid for all $i_0=1,\dots,\nv$ and $k_0\in\zsp$.  

The other properties are proved similarly to the proof of Theorem~\ref{d1evcov}.
\end{proof}

For each $n\in\zp$, let $\mat_n$ be the set of matrices of 
size~$m\times (n+1)$ with nonnegative integer entries. 
For a matrix $\gamma\in\mat_n$, 
its entries are denoted by $\gamma_{ik}\in\zp$, 
where $i=1,\dots,m$ and $k=0,\dots,n$. 
Let $U^\gamma$ be the following product 
\begin{equation}
\label{ugamma}
U^\gamma=\prod_{\substack{i=1,\dots,m,\\k=0,\dots,n}}
\big(u^i_k-a^i_k\big)^{\gamma_{ik}}.
\end{equation}

For each $n\in\zsp$, $i_0\in\{1,\dots,\nv\}$, and $k_0\in\{1,\dots,n\}$, 
denote by $M_{i_0,k_0}^n\subset\mat_n$
the subset of matrices $\al$ satisfying the following conditions
\begin{gather}
\lb{ali0k0}
\al_{i_0k_0}=1,\quad\forall\,k>k_0\quad\forall\,i\quad\al_{ik}=0,
\quad\forall\,i_1\neq i_0\quad\al_{i_1k_0}=0,\quad
\forall\,i_2>i_0\quad\al_{i_2,k_0-1}=0.
\end{gather}
In other words, for each $k>k_0$ the $k$-th column of any matrix $\al\in M_{i_0,k_0}^n$ is zero,
the $k_0$-th column contains only one nonzero entry $\al_{i_0k_0}=1$, and in the $(k_0-1)$-th column
one has $\al_{i_2,k_0-1}=0$ for all $i_2>i_0$.

Set also $M_{i_0,k_0}^0=\varnothing$ for all $i_0$, $k_0$. 

Consider a covering of order $\le\ocn$ given by vector fields $D_x+A$, $D_t+B$.
Let $b\in W$. 
We are going to study the structure of this covering  
on a neighborhood of the point $(a,b)\in\CE\times W$. 

Recall that the vector fields $A$, $B$ satisfy~\er{orderna},~\er{ordernb} 
and are analytic, according to the convention from Section~\ref{subs-conv}.
Therefore, taking a sufficiently small neighborhood of $(a,b)$,
we can assume that $A$ and $B$ are represented as absolutely convergent series
\begin{gather}
\label{aser}
A=\sum_{\al\in \mat_\ocn,\ l_1,l_2\in\zp}(x-x_0)^{l_1}(t-t_0)^{l_2}\cdot U^\al\cdot A^{l_1,l_2}_\al,\\
\lb{bser}
B=\sum_{\beta\in \mat_{\ocn+\eo-1},\ l_1,l_2\in\zp}(x-x_0)^{l_1}(t-t_0)^{l_2}\cdot U^\beta\cdot B^{l_1,l_2}_\beta,
\end{gather}
where $A^{l_1,l_2}_\al$, $B^{l_1,l_2}_\beta$ are vector fields on an open subset of~$W$.

\begin{remark}
\lb{abcoef0}
According to Theorem~\ref{evcov}, after a suitable gauge transformation we get
properties~\er{gaukak},~\er{gbxx0}, and~\er{gd=0} for all $i_0=1,\dots,\nv$ and $k_0\in\zsp$.
Using formulas~\er{aser},~\er{bser}, one obtains that these properties are equivalent to 
\beq
\lb{ab000}
A^{l_1,l_2}_0=B^{0,l_2}_0=0,\quad
A^{l_1,l_2}_{\tilde\al}=0,\quad\tilde\al\in M^\ocn_{i_0,k_0},\quad i_0=1,\dots,\nv,\quad k_0=1,\dots,\ocn,
\quad l_1,l_2\in\zp.
\ee
 
\end{remark}

\subsection{The algebras $\fdn^\ocn(\CE,a)$}
\lb{deffdoc}

Let $\ocn\in\zp$. Consider a point $a\in\CE$ given by~\er{pointev}.

\begin{remark}  
\label{inform}
The main idea of the definition of the Lie algebra $\fdn^\ocn(\CE,a)$  
can be informally outlined as follows. 
According to Theorem~\ref{evcov} and Remark~\ref{abcoef0}, 
any covering of order $\le\ocn$ is locally gauge equivalent 
to a covering given by vector fields $A$, $B$ that are of the form~\er{aser},~\er{bser} 
and satisfy~\er{gc}, \er{ab000}.

To define $\fdn^\ocn(\CE,a)$, we regard $A^{l_1,l_2}_\al$, $B^{l_1,l_2}_\beta$ from~\er{aser},~\er{bser} 
as abstract symbols. 
By definition, the algebra $\fdn^\ocn(\CE,a)$ is generated by the symbols $A^{l_1,l_2}_\al$, $B^{l_1,l_2}_\beta$ 
for ${\al\in\mat_\ocn}$, ${\be\in\mat_{\ocn+\eo-1}}$, $l_1,l_2\in\zp$.
Relations for these generators are provided by equations~\er{gc}, \er{ab000}.
The details of this construction are presented below.


\end{remark}

Let $\frl$ be the free Lie algebra generated 
by the symbols $\fla^{l_1,l_2}_\al$, $\flb^{l_1,l_2}_\beta$ for 
${\al\in\mat_\ocn}$, ${\be\in\mat_{\ocn+\eo-1}}$, $l_1,l_2\in\zp$.
In particular, we have
$$
\fla^{l_1,l_2}_\al\in\frl,\quad\ 
\flb^{l_1,l_2}_\beta\in\frl,\quad\ 
\big[\fla^{l_1,l_2}_\al,\flb^{l_1,l_2}_\beta\big]\in\frl\qquad 
\forall\,\al\in\mat_\ocn,\qquad \forall\,\be\in\mat_{\ocn+\eo-1},\qquad 
\forall\,l_1,l_2\in\zp.
$$
Consider the following formal power series with coefficients in~$\frl$
\begin{gather*}
\notag
\fla=\sum_{\al\in \mat_\ocn,\ l_1,l_2\in\zp}(x-x_0)^{l_1}(t-t_0)^{l_2}\cdot U^\al
\cdot \fla^{l_1,l_2}_\al,\\
\flb=\sum_{\beta\in \mat_{\ocn+\eo-1},\ l_1,l_2\in\zp}
(x-x_0)^{l_1}(t-t_0)^{l_2}\cdot U^\beta\cdot\flb^{l_1,l_2}_\beta.
\end{gather*}

Set 
\begin{gather}
\lb{dxflb}
D_x(\flb)=\sum_{\beta\in \mat_{\ocn+\eo-1},\ l_1,l_2\in\zp}
D_x\big((x-x_0)^{l_1}(t-t_0)^{l_2}U^\beta\big)\cdot\flb^{l_1,l_2}_\beta,\\
\lb{dtfla}
D_t(\fla)=\sum_{\al\in \mat_\ocn,\ l_1,l_2\in\zp}D_t\big((x-x_0)^{l_1}(t-t_0)^{l_2}U^\al\big)
\cdot \fla^{l_1,l_2}_\al,\\
\label{lieab}
[\fla,\flb]=\sum_{\substack{\al\in\mat_\ocn,\ \beta\in \mat_{\ocn+\eo-1},\\
l_1,l_2,l'_1,l'_2\in\zp}} 
(x-x_0)^{l_1+l'_1}(t-t_0)^{l_2+l'_2}\cdot 
U^\al\cdot U^{\beta}\cdot\big[\fla^{l_1,l_2}_\al,\flb^{l'_1,l'_2}_\beta\big].
\end{gather}
For any $\al\in\mat_\ocn$, $\beta\in \mat_{\ocn+\eo-1}$, $l_1,l_2\in\zp$, 
the expressions $D_x\big((x-x_0)^{l_1}(t-t_0)^{l_2}U^\beta\big)$ 
and $D_t\big((x-x_0)^{l_1}(t-t_0)^{l_2}U^\al\big)$ 
are functions of the variables $x$, $t$, $u^i_k$. 
Taking the corresponding Taylor series at the point~\eqref{pointev}, 
we regard these expressions as power series. 

Then~\er{dxflb},~\er{dtfla},~\er{lieab} are formal power series with coefficients in~$\frl$, 
and we have 
\begin{equation*}
D_x(\flb)-D_t(\fla)+[\fla,\flb]=
\sum_{\gamma\in \mat_{\ocn+\eo},\ l_1,l_2\in\zp}
(x-x_0)^{l_1}(t-t_0)^{l_2}\cdot U^\gamma\cdot\flz^{l_1,l_2}_\gamma
\end{equation*}
for some elements $\flz^{l_1,l_2}_\gamma\in\frl$. 

Let $\frid\subset\frl$ be the ideal generated by the elements
\begin{gather*}
\flz^{l_1,l_2}_\gamma,\qquad\fla^{l_1,l_2}_0,\qquad 
\flb^{0,l_2}_0,\qquad\gamma\in\mat_{\ocn+\eo},\qquad l_1,l_2\in\zp,\\
\fla^{l_1,l_2}_{\tilde\al},\qquad{\tilde\al}\in M^\ocn_{i_0,k_0},\qquad i_0=1,\dots,\nv,\qquad k_0=1,\dots,\ocn,
\qquad l_1,l_2\in\zp.
\end{gather*}
Set $\fdn^\ocn(\CE,a)=\frl/\frid$. 
Consider the natural homomorphism  
$\rho\cl\frl\to\frl/\frid=\fdn^\ocn(\CE,a)$ and set 
$$
\ga^{l_1,l_2}_\al=\rho\big(\fla^{l_1,l_2}_\al\big),\qquad\qquad 
\gb^{l_1,l_2}_\beta=\rho\big(\flb^{l_1,l_2}_\beta\big).
$$
The definition of~$\frid$ implies that the power series 
\begin{gather}
\label{gasumxt}
\ga=\sum_{\al\in \mat_\ocn,\ l_1,l_2\in\zp}(x-x_0)^{l_1}(t-t_0)^{l_2}\cdot U^\al
\cdot \ga^{l_1,l_2}_\al,\\
\label{gbsumxt}
\gb=\sum_{\beta\in \mat_{\ocn+\eo-1},\ l_1,l_2\in\zp}
(x-x_0)^{l_1}(t-t_0)^{l_2}\cdot U^\beta\cdot\gb^{l_1,l_2}_\beta.
\end{gather}
satisfy 
\beq
\lb{xgbtga}
D_x(\gb)-D_t(\ga)+[\ga,\gb]=0.
\ee

\begin{remark}
\lb{rem_fdpgen}
The Lie algebra $\fdn^\ocn(\CE,a)$ can be described in terms 
of generators and relations as follows. 

Equation~\er{xgbtga} is equivalent to some Lie algebraic relations for $\ga^{l_1,l_2}_\al$, $\gb^{l_1,l_2}_\beta$.

The algebra $\fdn^\ocn(\CE,a)$ is given by the generators $\ga^{l_1,l_2}_\al$, $\gb^{l_1,l_2}_\beta$,  
the relations arising from~\er{xgbtga}, and the following relations 
\beq
\lb{gagb00}
\ga^{l_1,l_2}_0=\gb^{0,l_2}_0=0,\quad
\ga^{l_1,l_2}_{\tilde\al}=0,\quad\tilde\al\in M^\ocn_{i_0,k_0},\quad i_0=1,\dots,\nv,\quad k_0=1,\dots,\ocn,
\quad l_1,l_2\in\zp.
\ee
\end{remark}


Recall that an \emph{action} of a Lie algebra $\bl$ on a manifold $W$
is a homomorphism from $\bl$ to the Lie algebra $\ve(W)$ of vector fields on $W$.

Let $W_1$, $W_2$ be manifolds, and $\rho_i\cl\bl\to\ve(W_i)$ be an action of $\bl$ on $W_i$ for $i=1,2$.
A \emph{morphism} connecting the actions 
$\rho_i\cl\bl\to\ve(W_i)$, $i=1,2$, is a map
$\vf\cl W_1\to W_2$ such that for any $Y\in\bl$ 
one has $\vf_*(\rho_1(Y))=\rho_2(Y)$, where $\vf_*$ is the differential of $\vf$.

Suppose that we have an action of $\fdn^\ocn(\CE,a)$ on a manifold $W$ given by
$$
\ga^{l_1,l_2}_{\al}\mapsto A^{l_1,l_2}_{\al}\in\ve(W),\qquad\qquad
\gb^{l_1,l_2}_{\beta}\mapsto B^{l_1,l_2}_{\beta}\in\ve(W)
$$
such that the corresponding 
power series~\er{aser},~\er{bser} are absolutely convergent
on a neighborhood of~$a$.  Then from~\er{xgbtga} it follows that~\er{aser},~\er{bser} satisfy~\er{gc}
and, therefore, determine a covering.

Combining this construction with Theorem~\ref{evcov} and Remark~\ref{abcoef0}, we obtain the following result.
\begin{theorem}
Any covering of order $\le\ocn$ on a neighborhood of $a\in\CE$
is locally gauge equivalent to the covering arising from
an action of the Lie algebra $\fdn^\ocn(\CE,a)$.

For a fixed covering of order $\le\ocn$, the corresponding action of~$\fdn^\ocn(\CE,a)$
is defined uniquely up to a local isomorphism.
\end{theorem}

Suppose that $\ocn\ge 1$.
Since any covering of order $\le\ocn-1$ is at the same time of 
order $\le\ocn$,
we have the surjective homomorphism $\fdn^{\ocn}(\CE,a)\to\fdn^{\ocn-1}(\CE,a)$ that maps the generators
\begin{gather*}
\ga^{l_1,l_2}_{\al},\qquad\exists\,i\qquad\al_{i,p}\neq 0,\\
\gb^{l_1,l_2}_{\beta},\qquad\exists\,i'\qquad{\beta}_{i',\ocn+\eo-1}\neq 0,
\end{gather*}
to zero and maps 
the other generators of $\fdn^{\ocn}(\CE,a)$ to the corresponding 
generators of $\fdn^{\ocn-1}(\CE,a)$.

Thus we obtain the following sequence of surjective homomorphisms of Lie algebras
\beq
\lb{fdnn-1}
\dots\to\fdn^{\ocn}(\CE,a)\to\fdn^{\ocn-1}(\CE,a)\to\dots\to\fdn^1(\CE,a)\to\fdn^0(\CE,a).
\ee



\section{The homomorphisms $\fdn^\ocn(\CE,a)\to\fdn^{\ocn-1}(\CE,a)$ and $\fdn^\ocn(\CE,a)\to\fdn^{0}(\CE,a)$ 
for KdV type equations}
\lb{seckdvtype}

In this section we study the algebras~\er{fdnn-1} for equations of the form 
\beq 
\lb{kdveq}
u_t=u_{xxx}+f(u,u_x),
\ee 
where $f$ is an arbitrary function. 

Set $u_0=u$ and $u_k=\dfrac{\pd^k u}{\pd x^k}$ for $k\in\zsp$. 
Let $\CE$ be the infinite prolongation of equation~\er{kdveq}. 
Then $\CE$ is the infinite-dimensional manifold with the coordinates $x$, $t$, $u_k$, $k\in\zp$. 

For equation~\er{kdveq}, the total derivative operators~\er{evdxdt} are 
\beq
\label{kdvdxdt}
  D_x=\frac{\pd}{\pd x}+\sum_{k\ge 0} u_{k+1}\frac{\pd}{\pd u_k},\qquad\qquad
  D_t=\frac{\pd}{\pd t}+\sum_{k\ge 0} 
  D_x^k\big(u_3+f(u_0,u_1)\big)\frac{\pd}{\pd u_k}.
\ee
Consider an arbitrary point $a\in\CE$ given by  
\begin{equation}
\lb{pointkdv}
a=(x=x_0,\,t=t_0,\,u_k=a_k)\,\in\,\CE,\qquad\qquad x_0,\,t_0,\,a_k\in\Com,\qquad k\in\zp.
\end{equation}
Since equation~\er{kdveq} is invariant with respect to the change of variables
$x\mapsto x-x_0$, $t\mapsto t-t_0$, we can assume $x_0=t_0=0$.


Let $\oc\in\zsp$.
According to Section~\ref{deffdoc}, 
the algebra $\fdn^\ocn(\CE,a)$ is described as follows.
Consider formal power series 
\begin{gather}
\label{kdvgaser}
\ga=\sum_{l_1,l_2,i_0,\dots,i_\ocn\ge 0} 
x^{l_1} t^{l_2}(u_0-a_0)^{i_0}\dots(u_\ocn-a_\ocn)^{i_\ocn}\cdot
\ga^{l_1,l_2}_{i_0\dots i_\ocn},\\
\lb{kdvgbser}
\gb=\sum_{l_1,l_2,j_0,\dots,j_{\ocn+2}\ge 0} 
x^{l_1} t^{l_2}(u_0-a_0)^{j_0}\dots(u_{\ocn+2}-a_{\ocn+2})^{j_{\ocn+2}}\cdot
\gb^{l_1,l_2}_{j_0\dots j_{\ocn+2}}
\end{gather}
satisfying 
\begin{gather}
\label{relgakdv}
\ga^{l_1,l_2}_{i_0\dots i_\ocn}=0\quad \text{if}\quad \exists\,r\in\{1,\dots,\ocn\}\quad \text{such that}\quad 
i_r=1,\quad i_n=0\quad \forall\,n>r,\\
\lb{kdvaukak}
\ga^{l_1,l_2}_{0\dots 0}=0\qquad\qquad\forall\,l_1,l_2\in\zp,\\
\lb{kdvbxx0}
\gb^{0,l_2}_{0\dots 0}=0\qquad\qquad\forall\,l_2\in\zp.
\end{gather}
Then $\ga^{l_1,l_2}_{i_0\dots i_\ocn}$, $\gb^{l_1,l_2}_{j_0\dots j_{\ocn+2}}$ 
are generators of the algebra $\fdn^\ocn(\CE,a)$, and the equation 
\beq
\label{gckdv}
D_x(\gb)-D_t(\ga)+[\ga,\gb]=0
\ee
provides relations for these generators (in addition to relations~\er{relgakdv},~\er{kdvaukak},~\er{kdvbxx0}). 

Note that condition~\er{relgakdv} is equivalent to 
\beq
\label{kdvd=0}
\frac{\pd}{\pd u_s}(\ga)\,\,\bigg|_{u_k=a_k,\ k\ge s}=0\qquad\qquad
\forall\,s\in\zsp.
\ee

Using~\er{kdvdxdt}, we can rewrite equation~\er{gckdv} as
\beq
\lb{kdvqzcr}
\frac{\pd}{\pd x}(\gb)+\sum_{k=0}^{\ocn+2} u_{k+1}\frac{\pd}{\pd u_k}(\gb)
-\frac{\pd}{\pd t}(\ga)-\sum_{k=0}^\ocn\Big(u_{k+3}+D_x^k\big(f(u_0,u_1)\big)\Big)\frac{\pd}{\pd u_k}(\ga)
+[\ga,\gb]=0.
\ee

\begin{proposition}
\lb{lemgenfdq}
The elements 
\beq
\lb{gal1alprop}
\ga^{l_1,0}_{i_0\dots i_\ocn},\qquad\qquad l_1,i_0,\dots,i_\ocn\in\zp,
\ee
generate the algebra $\fdn^\ocn(\ce,a)$.
\end{proposition}
\begin{proof}
For each $l\in\zp$, 
denote by $\mg_l\subset \fdn^{\ocn}(\CE,a)$ the subalgebra generated by 
the elements $\ga^{l_1,l_2}_{i_0\dots i_\ocn}$ with $l_2\le l$.
To prove Proposition~\ref{lemgenfdq}, we need some lemmas.

\begin{lemma}
\label{gbllalgl}
Let $l_1,l_2,j_0,\dots,j_{\ocn+2}\in\zp$ be such that $j_0+\dots+j_{\ocn+2}>0$. 
Then $\gb^{l_1,l_2}_{j_0\dots j_{\ocn+2}}\in\mg_{l_2}$.
\end{lemma}
\begin{proof}
For any $j_0,\dots,j_{\ocn+2}\in\zp$ satisfying $j_0+\dots+j_{\ocn+2}>0$, 
denote by $\rho(j_0,\dots,j_{\ocn+2})$ the maximal integer $r\in\{0,1,\dots,\ocn+2\}$ 
such that $j_r\neq 0$. 

Differentiating~\eqref{kdvqzcr} with respect to $u_{\ocn+3}$, we obtain 
\beq
\lb{uq2buqa}
\dfrac{\pd}{\pd u_{\ocn+2}}(\gb)=\dfrac{\pd}{\pd u_\ocn}(\ga),
\ee
which implies $\gb^{l_1,l_2}_{j_0\dots j_{\ocn+2}}\in\mg_{l_2}$ 
for all $l_1,l_2,j_0,\dots,j_{\ocn+2}\in\zp$ obeying $\rho(j_0,\dots,j_{\ocn+2})=\ocn+2$.

Let $n\in\{0,1,\dots,\ocn+1\}$ be such that  
\beq
\lb{assumgbn}
\gb^{l_1,l_2}_{j'_0\dots j'_{\ocn+2}}\in\mg_{l_2}\quad
\text{for all $l_1,l_2,j'_0,\dots,j'_{\ocn+2}\in\zp$ satisfying $\rho(j'_0,\dots,j'_{\ocn+2})>n$}.
\ee  
We are going to show that 
$\gb^{l_1,l_2}_{\tilde\jmath_0\dots\tilde\jmath_{\ocn+2}}\in\mg_{l_2}$ 
for all $l_1,l_2,\tilde\jmath_0,\dots,\tilde\jmath_{\ocn+2}\in\zp$ 
satisfying 
$\rho(\tilde\jmath_0,\dots,\tilde\jmath_{\ocn+2})=n$.  

For any power series $C$ of the form 
$$
C=\sum_{l_1,l_2,d_0,\dots,d_k\ge 0} 
x^{l_1} t^{l_2}(u_0-a_0)^{d_0}\dots(u_k-a_k)^{d_k}\cdot
C^{l_1,l_2}_{d_0\dots d_k},\qquad\quad C^{l_1,l_2}_{d_0\dots d_k}\in\fdn^\ocn(\ce,a),
$$
set 
$$
\ds(C)=\Big(\frac{\pd}{\pd u_{n+1}}(C)\Big)\,\Bigg|_{u_k=a_k,\ k\ge n+1}.
$$
That is, in order to obtain $\ds(C)$, we differentiate $C$ with respect to $u_{n+1}$ 
and then substitute $u_k=a_k$ for all $k\ge n+1$.

Equation~\er{kdvd=0} implies
\beq
\lb{dsga0}
\ds\Big(\frac{\pd}{\pd t}(\ga)\Big)=0.
\ee
Combining~\er{kdvqzcr} with~\er{dsga0}, we get 
\beq
\lb{dsdxgbsum}
\ds\big(D_x(\gb)\big)=
\ds\bigg(\sum_{k=0}^\ocn\Big(u_{k+3}+D_x^k\big(f(u_0,u_1)\big)\Big)\frac{\pd}{\pd u_k}(\ga)\bigg)
-\ds\big([\ga,\gb]\big).
\ee
In equation~\er{dsdxgbsum}, we regard $f(u_0,u_1)$ as a power series, 
using the Taylor series of the function $f(u_0,u_1)$ at the point~\er{pointkdv}.

Using~\er{kdvgbser}, one obtains
\begin{multline}
\label{dsdxgb}
\ds\big(D_x(\gb)\big)=
\sum_{\substack{l_1,l_2,j_0,\dots,j_{\ocn+2}\ge 0,\\
\rho(j_0,\dots,j_{\ocn+2})=n}}
j_{n}x^{l_1}t^{l_2}
(u_0-a_0)^{j_0}\dots(u_{n-1}-a_{n-1})^{j_{n-1}}(u_{n}-a_{n})^{j_{n}-1}
\gb^{l_1,l_2}_{j_0\dots j_{\ocn+2}}+\\
+\ds\Bigg(\sum_{\substack{l_1,l_2,j_0,\dots,j_{\ocn+2}\ge 0,\\
\rho(j_0,\dots,j_{\ocn+2})>n}}
t^{l_2} 
D_x\Big(x^{l_1}(u_0-a_0)^{j_0}\dots(u_{\ocn+2}-a_{\ocn+2})^{j_{\ocn+2}}\Big)
\cdot
\gb^{l_1,l_2}_{j_0\dots j_{\ocn+2}}\Bigg).
\end{multline}
From~\er{kdvd=0} it follows that $\ds(\ga)=0$, which yields 
\begin{multline}
\label{dsgagb}
\ds\big([\ga,\gb]\big)=\Big[\ds(\ga),\,\gb\,\Big|_{u_k=a_k,\ k\ge n+1}\Big]
+\Big[\ga\,\Big|_{u_k=a_k,\ k\ge n+1},\,\ds(\gb)\Big]=\\
=\Bigg[\ga\,\Big|_{u_k=a_k,\ k\ge n+1},\,
\ds\Bigg(\sum_{\substack{l_1,l_2,j_0,\dots,j_{\ocn+2}\ge 0,\\
\rho(j_0,\dots,j_{\ocn+2})>n}}
x^{l_1} t^{l_2}(u_0-a_0)^{j_0}\dots(u_{\ocn+2}-a_{\ocn+2})^{j_{\ocn+2}}\cdot
\gb^{l_1,l_2}_{j_0\dots j_{\ocn+2}}\Bigg)\Bigg].
\end{multline}
In view of~\er{dsdxgb},~\er{dsgagb},
for any $l_1,l_2,\tilde\jmath_0,\dots,\tilde\jmath_{\ocn+2}\in\zp$ satisfying  
$\rho(\tilde\jmath_0,\dots,\tilde\jmath_{\ocn+2})=n$ 
the element $\gb^{l_1,l_2}_{\tilde\jmath_0\dots\tilde\jmath_{\ocn+2}}$ 
appears only once on the left-hand side of~\er{dsdxgbsum} 
and does not appear on the right-hand side of~\er{dsdxgbsum}.

Combining~\er{dsdxgbsum},~\er{dsdxgb},~\er{dsgagb},
we obtain that the element $\gb^{l_1,l_2}_{\tilde\jmath_0\dots\tilde\jmath_{\ocn+2}}$ 
is equal to a linear combination of elements of the form 
\beq
\label{elemgagb}
\ga^{l_1',l_2'}_{i_0\dots i_\ocn},\quad
\gb^{\hat{l}_1,\hat{l}_2}_{\hat{\jmath}_0\dots\hat{\jmath}_{\ocn+2}},\quad
\Big[\ga^{l_1',l_2'}_{i_0\dots i_\ocn},
\gb^{\hat{l}_1,\hat{l}_2}_{\hat{\jmath}_0\dots\hat{\jmath}_{\ocn+2}}\Big],\quad
l_2'\le l_2,\quad \hat{l}_2\le l_2,\quad 
\rho(\hat{\jmath}_0,\dots,\hat{\jmath}_{\ocn+2})>n.
\ee
Obviously, for any $\hat{l}_2\le l_2$ one has $\mg_{\hat{l}_2}\subset\mg_{l_2}$.
Taking into account assumption~\er{assumgbn}, 
we obtain that the elements~\er{elemgagb} belong to $\mg_{{l}_2}$. 
Hence $\gb^{l_1,l_2}_{\tilde\jmath_0\dots\tilde\jmath_{\ocn+2}}\in\mg_{{l}_2}$.

The proof is completed by induction.
\end{proof}

\begin{lemma}
\label{gb00mgl}
For all $l_1,l_2\in\zp$, one has $\gb^{l_1,l_2}_{0\dots 0}\in\mg_{l_2}$.
\end{lemma}
\begin{proof}
According to~\er{kdvbxx0}, we have $\gb^{0,l_2}_{0\dots 0}=0$. 
Therefore, it is sufficient 
to prove $\gb^{l_1,l_2}_{0\dots 0}\in\mg_{l_2}$ for $l_1>0$.


Note that condition~\er{kdvaukak} implies 
\beq
\lb{gatga0}
\ga\Big|_{u_k=a_k,\ k\ge 0}=0,\qquad\qquad
\frac{\pd}{\pd t}(\ga)\,\bigg|_{u_k=a_k,\ k\ge 0}=0.
\ee
In view of~\er{kdvgbser}, one has  
\beq
\lb{pdxgb0}
\frac{\pd}{\pd x}(\gb)\,\bigg|_{u_k=a_k,\ k\ge 0}=
\sum_{l_1>0,\ l_2\ge 0}l_1x^{l_1-1} t^{l_2}\cdot\gb^{l_1,l_2}_{0\dots 0}.
\ee
Substituting $u_k=a_k$ for all $k\in\zp$ in~\er{kdvqzcr} and 
using~\er{gatga0},~\er{pdxgb0}, we get 
\begin{multline}
\lb{suml1gb0}
\sum_{l_1>0,\ l_2\ge 0} 
l_1x^{l_1-1} t^{l_2}\cdot\gb^{l_1,l_2}_{0\dots 0}=\\
=-\bigg(\sum_{k=0}^{\ocn+2} u_{k+1}\frac{\pd}{\pd u_k}(\gb)\bigg)
\,\bigg|_{u_k=a_k,\ k\ge 0}
+\bigg(\sum_{k=0}^\ocn\Big(u_{k+3}+D_x^k\big(f(u_0,u_1)\big)\Big)
\frac{\pd}{\pd u_k}(\ga)\bigg)\,\bigg|_{u_k=a_k,\ k\ge 0}.
\end{multline}
Combining~\er{kdvgaser},~\er{kdvgbser},~\er{suml1gb0}, 
we see that for any $l_1>0$ and $l_2\ge 0$ the element $\gb^{l_1,l_2}_{0\dots 0}$
is equal to a linear combination of elements of the form 
\beq
\label{gagabj1}
\ga^{l'_1,l_2}_{i_0\dots i_\ocn},\qquad
\gb^{l'_1,l_2}_{j_0\dots j_{\ocn+2}},\qquad\quad 
j_0+\dots+j_{\ocn+2}=1.
\ee
According to Lemma~\ref{gbllalgl} and the definition of $\mg_{l_2}$, 
the elements~\er{gagabj1} belong to $\mg_{l_2}$. 
Thus $\gb^{l_1,l_2}_{0\dots 0}\in\mg_{l_2}$.
\end{proof}

\begin{lemma}
\label{gallmg}
For all $l_1,l,i_0,\dots,i_\ocn\in\zp$, we have 
$\ga^{l_1,l+1}_{i_0\dots i_\ocn}\in\mg_l$.
\end{lemma}
\begin{proof}
Using~\er{kdvgaser}, we can rewrite equation~\er{kdvqzcr} as
\begin{multline*}
\sum_{l_1,l,i_0,\dots,i_\ocn\ge 0}(l+1) 
x^{l_1} t^{l}(u_0-a_0)^{i_0}\dots(u_\ocn-a_\ocn)^{i_\ocn}\cdot\ga^{l_1,l+1}_{i_0\dots i_\ocn}=\\
=\frac{\pd}{\pd x}(\gb)+\sum_{k=0}^{\ocn+2} u_{k+1}\frac{\pd}{\pd u_k}(\gb)
-\sum_{k=0}^\ocn \Big(u_{k+3}+D_x^k\big(f(u_0,u_1)\big)\Big)
\frac{\pd}{\pd u_k}(\ga)+[\ga,\gb].
\end{multline*}
This implies that $\ga^{l_1,l+1}_{i_0\dots i_\ocn}$ 
is equal to a linear combination of elements of the form 
\beq
\label{elemgall}
\ga^{\hat{l}_1,\hat{l}_2}_{\hat{\imath}_0\dots\hat{\imath}_\ocn},\qquad
\gb^{\tilde{l}_1,\tilde{l}_2}_{\tilde{\jmath}_0\dots\tilde{\jmath}_{\ocn+2}},\qquad
\Big[\ga^{\hat{l}_1,\hat{l}_2}_{\hat{\imath}_0\dots\hat{\imath}_\ocn},
\gb^{\tilde{l}_1,\tilde{l}_2}_{\tilde{\jmath}_0\dots\tilde{\jmath}_{\ocn+2}}\Big],\qquad
\hat{l}_2\le l,\qquad \tilde{l}_2\le l.
\ee
Using Lemmas~\ref{gbllalgl},~\ref{gb00mgl} and the condition~$\tilde{l}_2\le l$, we get 
$\gb^{\tilde{l}_1,\tilde{l}_2}_{\tilde{\jmath}_0\dots\tilde{\jmath}_{\ocn+2}}\in\mg_{\tilde{l}_2}\subset\mg_{l}$. 
Therefore, the elements~\er{elemgall} belong to $\mg_{l}$. 
Hence $\ga^{l_1,l+1}_{i_0\dots i_\ocn}\in\mg_{l}$.
\end{proof}

Return to the proof of Proposition~\ref{lemgenfdq}.
According to Lemmas~\ref{gbllalgl},~\ref{gb00mgl} and the definition of~$\mg_{l}$, 
we have 
$\ga^{l_1,l_2}_{i_0\dots i_\ocn},\gb^{l_1,l_2}_{j_0\dots j_{\ocn+2}}\in\mg_{l_2}$ 
for all $l_1,l_2,i_0,\dots i_\ocn,j_0,\dots,j_{\ocn+2}\in\zp$. 
Lemma~\ref{gallmg} implies that 
\beq
\notag
\mg_{l_2}\subset\mg_{l_2-1}\subset\mg_{l_2-2}\subset\dots\subset\mg_0.
\ee
Therefore, $\fdn^{\ocn}(\CE,a)$ is equal to $\mg_0$, which is generated by the elements~\er{gal1alprop}.
\end{proof}

From~\er{uq2buqa} it follows that $\gb$ is of the form
\begin{equation}
\label{buq2a}
\gb=u_{\ocn+2}\frac{\pd}{\pd u_\ocn}(\ga)+\gb_0(x,t,u_0,\dots,u_{\ocn+1}),
\end{equation}
where $\gb_0(x,t,u_0,\dots,u_{\ocn+1})$ is a power series in the variables 
$x$, $t$, $u_0-a_0,\dots,u_{\ocn+1}-a_{\ocn+1}$. 

Differentiating~\eqref{kdvqzcr} with respect to $u_{\ocn+2}$, $u_{\ocn+1}$ 
and using~\er{buq2a}, one gets 
\begin{equation*}
\frac{\pd^2}{\pd u_{\ocn}\pd u_{\ocn}}(\ga)+\frac{\pd^2}{\pd u_{\ocn+1}\pd u_{\ocn+1}}(\gb_0)=0.
\end{equation*}
Therefore, $\gb_0=\gb_0(x,t,u_0,\dots,u_{\ocn+1})$ is of the form 
\begin{equation}
\label{buq2b}
\gb_0=-\frac12(u_{\ocn+1})^2\frac{\pd^2}{\pd u_{\ocn}\pd u_{\ocn}}(\ga)+
u_{\ocn+1}\gb_{01}(x,t,u_0,\dots,u_{\ocn})+\gb_{00}(x,t,u_0,\dots,u_{\ocn}),
\end{equation}
where $\gb_{0i}(x,t,u_0,\dots,u_{\ocn})$ is a power series 
in the variables $x$, $t$, $u_0-a_0,\dots,u_{\ocn}-a_{\ocn}$ for $i=0,1$.

Applying the operator~$\dfrac{\pd^3}{\pd u_{\ocn+1}\pd u_{\ocn+1}\pd u_{\ocn+1}}$ 
to equation~\er{kdvqzcr} and using~\er{buq2a},~\er{buq2b}, 
we get 
$$
\dfrac{\pd^3}{\pd u_{\ocn}\pd u_{\ocn}\pd u_{\ocn}}(\ga)=0.
$$
Hence $\ga$ is of the form
\beq
\label{ga210}
\ga=(u_{\ocn}-a_\ocn)^2\ga_2(x,t,u_0,\dots,u_{\ocn-1})+
(u_{\ocn}-a_\ocn)\ga_1(x,t,u_0,\dots,u_{\ocn-1})+\ga_0(x,t,u_0,\dots,u_{\ocn-1}),
\ee
where $\ga_j(x,t,u_0,\dots,u_{\ocn-1})$ is a power series 
in the variables $x$, $t$, $u_0-a_0,\dots,u_{\ocn-1}-a_{\ocn-1}$ for $j=0,1,2$.

Recall that we assume $\ocn\ge 1$.
Equation~\er{kdvd=0} for $s=\ocn$ yields 
\beq
\label{uqga10}
\ga_1(x,t,u_0,\dots,u_{\ocn-1})=0.
\ee
Combining~\er{buq2a},~\er{buq2b},~\er{ga210},~\er{uqga10}, we get
\begin{multline}
\label{gblong}
\gb=2u_{\ocn+2}(u_{\ocn}-a_\ocn)\ga_2(x,t,u_0,\dots,u_{\ocn-1})-(u_{\ocn+1})^2\ga_2(x,t,u_0,\dots,u_{\ocn-1})+\\
+u_{\ocn+1}\gb_{01}(x,t,u_0,\dots,u_{\ocn})+\gb_{00}(x,t,u_0,\dots,u_{\ocn}).
\end{multline}

Applying the operator~$\dfrac{\pd^2}{\pd u_{\ocn+1}\pd u_{\ocn+1}}$ to equation~\er{kdvqzcr}, 
one gets
\beq
\label{uq1uq1}
-2D_x(\ga_2)+2\frac{\pd}{\pd u_\ocn}(\gb_{01})-2[\ga_0,\ga_2]=0.
\ee
Differentiating~\er{uq1uq1} with respect to~$u_{\ocn}$, we obtain 
\beq
\label{uq1uq1uq}
-2\frac{\pd}{\pd u_{\ocn-1}}(\ga_2)+2\frac{\pd^2}{\pd u_\ocn\pd u_\ocn}(\gb_{01})=0.
\ee
Applying the operator~$\dfrac{\pd^3}{\pd u_{\ocn}\pd u_{\ocn}\pd u_{\ocn+2}}$ to
equation~\er{kdvqzcr}, one gets
\beq
\lb{uoc1a2}
4\frac{\pd}{\pd u_{\ocn-1}}(\ga_2)+
\frac{\pd^2}{\pd u_{\ocn}\pd u_{\ocn}}(\gb_{01})-2\frac{\pd}{\pd u_{\ocn-1}}(\ga_2)=0
\ee
Equations~\er{uq1uq1uq},~\er{uoc1a2} imply 
\beq
\label{uqqqa1}
\frac{\pd}{\pd u_{\ocn-1}}\big(\ga_2(x,t,u_0,\dots,u_{\ocn-1})\big)=0.
\ee
Applying the operator~$\dfrac{\pd^2}{\pd u_{\ocn}\pd u_{\ocn+2}}$ to equation~\er{kdvqzcr} and using~\er{uqqqa1}, we get 
\beq
\lb{dxga2uoc}
2D_x(\ga_2)+\frac{\pd}{\pd u_\ocn}(\gb_{01})+2[\ga_0,\ga_2]=0.
\ee
Combining~\er{dxga2uoc} with~\er{uq1uq1}, we obtain 
\beq
\lb{dxga2g0}
D_x(\ga_2)+[\ga_0,\ga_2]=0.
\ee
\begin{lemma}
\lb{ga2uk}
One has
\beq
\lb{ukga2}
\frac{\pd}{\pd u_k}(\ga_2)=0\qquad\qquad\forall\,k\in\zp.
\ee
\end{lemma}
\begin{proof}
Suppose that~\eqref{ukga2} does not hold. Let $k_0$ be the maximal integer 
such that $\dfrac{\pd}{\pd u_{k_0}}(\ga_2)\neq 0$. 

From~\er{uqqqa1} it follows that $k_0<\ocn-1$. 
Equation~\er{kdvd=0} for $s=k_0+1$ implies 
\beq
\lb{pduk0ga0}
\frac{\pd}{\pd u_{k_0+1}}(\ga_0)\,\bigg|_{u_k=a_k,\ k\ge k_0+1}=0.
\ee
Differentiating~\er{dxga2g0} with respect to~$u_{k_0+1}$, we obtain 
\beq
\lb{pdga2g0}
\frac{\pd}{\pd u_{k_0}}(\ga_2)+\Big[\frac{\pd}{\pd u_{k_0+1}}(\ga_0),\,\ga_2\Big]=0.
\ee
Substituting $u_k=a_k$ in~\er{pdga2g0} for all $k\ge k_0+1$ and 
using~\er{pduk0ga0}, one gets $\dfrac{\pd}{\pd u_{k_0}}(\ga_2)=0$, 
which contradicts to our assumption.				 
\end{proof}
From~\er{ukga2} it follows that equation~\er{dxga2g0} reads 
\beq
\lb{pdxga2g0}
\frac{\pd}{\pd x}(\ga_2)+[\ga_0,\ga_2]=0.
\ee
Note that condition~\er{kdvaukak} implies 
\beq
\lb{ga0ukak}
\ga_0\Big|_{u_k=a_k,\ k\ge 0}=0.
\ee
Substituting $u_k=a_k$ in~\er{pdxga2g0} for all $k\ge 0$ and using~\er{ukga2},~\er{ga0ukak}, 
we get 
\beq
\lb{pdxga20}
\frac{\pd}{\pd x}(\ga_2)=0.
\ee
Combining~\er{pdxga20} with~\er{pdxga2g0}, one obtains 
\beq
\lb{ga0ga20}
[\ga_2,\ga_0]=0.
\ee 

In view of~\er{kdvgaser},~\er{ga210}, we have 
\beq
\lb{ga0sum}
\ga_0=\sum_{l_1,l_2,i_0,\dots,i_{\ocn-1}\ge 0} 
x^{l_1} t^{l_2}(u_0-a_0)^{i_0}\dots(u_{\ocn-1}-a_{\ocn-1})^{i_{\ocn-1}}\cdot
\ga^{l_1,l_2}_{i_0\dots i_{\ocn-1}0}
\ee
According to~\er{kdvgaser},~\er{ga210},~\er{ukga2},~\er{pdxga20}, one has 
\beq
\lb{ga2sum}
\ga_2=\sum_{l\ge 0}t^l\cdot\tilde\ga^l,\qquad\qquad 
\tilde\ga^l=\ga^{0,l}_{0\dots 02}\in\fdn^\ocn(\CE,a).
\ee

Combining~\er{ga210},~\er{uqga10},~\er{ga0sum},~\er{ga2sum} with 
Proposition~\ref{lemgenfdq}, we obtain that the elements 
\beq
\lb{tga0ga}
\tilde\ga^0,\qquad\ga^{l_1,0}_{i_0\dots i_{\ocn-1}0},\qquad 
l_1,i_0,\dots,i_{\ocn-1}\in\zp,
\ee
generate the algebra $\fdn^\ocn(\CE,a)$.

Substituting $t=0$ in~\er{ga0ga20} and using~\er{ga0sum},~\er{ga2sum}, one gets 
\beq
\lb{tgagal1}
\big[\tilde\ga^0,
\ga^{l_1,0}_{i_0\dots i_{\ocn-1}0}\big]=0\qquad\forall\,l_1,i_0,\dots,i_{\ocn-1}\in\zp.
\ee
Since the elements~\er{tga0ga} generate the algebra $\fdn^\ocn(\CE,a)$, equation~\er{tgagal1} yields
\beq
\lb{tga0fdq}
\big[\tilde\ga^0,\,\fdn^\ocn(\CE,a)\big]=0.
\ee

\begin{lemma} One has
\beq
\lb{tgalfdq}
\big[\tilde\ga^l,\,\fdn^\ocn(\CE,a)\big]=0\qquad\qquad\forall\,l\in\zp.
\ee
\end{lemma}
\begin{proof}
We prove~\er{tgalfdq} by induction on $l$. 
The property $\big[\tilde\ga^0,\,\fdn^\ocn(\CE,a)\big]=0$ was obtained 
in~\er{tga0fdq}.
Let $n\in\zp$ be such that $\big[\tilde\ga^l,\,\fdn^\ocn(\CE,a)\big]=0$ 
for all $l\le n$. Since 
$\dfrac{\pd^{l}}{\pd t^{l}}(\ga_2)\,\bigg|_{t=0}=l!\cdot\tilde\ga^l$, 
we get
\beq
\lb{pdkga2}
\bigg[
\frac{\pd^{l}}{\pd t^{l}}(\ga_2)\,\bigg|_{t=0},\,
\frac{\pd^{m}}{\pd t^{m}}(\ga_0)\,\bigg|_{t=0}\bigg]=0\qquad\forall\,l\le n,
\qquad\forall\,m\in\zp.
\ee
Applying the operator $\dfrac{\pd^{n+1}}{\pd t^{n+1}}$ to equation~\er{ga0ga20}, 
substituting $t=0$, and using~\er{pdkga2}, one obtains 
\begin{multline*}
0=\frac{\pd^{n+1}}{\pd t^{n+1}}
\big([\ga_2,\ga_0]\big)\,\bigg|_{t=0}
=\sum_{k=0}^{n+1}\binom{n+1}{k}\cdot
\bigg[
\frac{\pd^{k}}{\pd t^{k}}(\ga_2)\,\bigg|_{t=0},\,
\frac{\pd^{n+1-k}}{\pd t^{n+1-k}}(\ga_0)\,\bigg|_{t=0}\bigg]=\\
=
\bigg[
\frac{\pd^{n+1}}{\pd t^{n+1}}(\ga_2)\,\bigg|_{t=0},\,
\ga_0\,\bigg|_{t=0}\bigg]=\\
=\bigg[
(n+1)!\cdot\tilde\ga^{n+1},\,
\sum_{l_1,i_0,\dots,i_{\ocn-1}} 
x^{l_1}(u_0-a_0)^{i_0}\dots(u_{\ocn-1}-a_{\ocn-1})^{i_{\ocn-1}}\cdot
\ga^{l_1,0}_{i_0\dots i_{\ocn-1}0}\bigg],
\end{multline*}
which implies 
\beq
\lb{tgan1ga0}
\big[\tilde\ga^{n+1},\,\ga^{l_1,0}_{i_0\dots i_{\ocn-1}0}\big]=0
\qquad\qquad\forall\,l_1,i_0,\dots,i_{\ocn-1}\in\zp.
\ee
Equation~\er{tga0fdq} yields
\beq
\lb{tga0tgan1}
\big[\tilde\ga^{0},\tilde\ga^{n+1}\big]=0.
\ee

Since the elements~\er{tga0ga} generate the algebra $\fdn^\ocn(\CE,a)$, 
from~\er{tgan1ga0},~\er{tga0tgan1} it follows that 
$\big[\tilde\ga^{n+1},\,\fdn^\ocn(\CE,a)\big]=0$. 
\end{proof}

\begin{theorem}
\lb{thcenter}
Let $\CE$ be the infinite prolongation of equation~\er{kdveq}. Let $a\in\CE$. 
For each $\ocn\in\zsp$, 
consider the homomorphism $\vf_\ocn\cl\fdn^\ocn(\CE,a)\to\fdn^{\ocn-1}(\CE,a)$ constructed 
in~\er{fdnn-1}. 
We have 
\beq
\lb{v1v2ker}
[v_1,v_2]=0\qquad\qquad\forall\,v_1\in\ker\vf_\ocn,\qquad\forall\,v_2\in\fdn^\ocn(\CE,a).
\ee
That is, the kernel of $\vf_\ocn$ is contained in the center of the Lie 
algebra $\fdn^\ocn(\CE,a)$.

For each $k\in\zsp$, let $\psi_k\colon\fdn^k(\CE,a)\to\fdn^{0}(\CE,a)$ 
be the composition of the homomorphisms
$$
\fdn^k(\CE,a)\to\fdn^{k-1}(\CE,a)\to\dots
\to\fdn^{1}(\CE,a)\to\fdn^{0}(\CE,a)
$$
from~\er{fdnn-1}. Then   
\beq
\lb{hhhhk}
[h_1,[h_2,\dots,[h_{k-1},[h_k,h_{k+1}]]\dots]]=0\qquad\qquad\forall\,h_1,\dots,h_{k+1}\in\ker\psi_k.
\ee
In particular, the kernel of $\psi_k$ is nilpotent.
\end{theorem}
\begin{proof}
Combining formulas~\er{ga210},~\er{uqga10},~\er{gblong},~\er{ga2sum} with the definition 
of the homomorphism $\vf_\ocn\cl\fdn^\ocn(\CE,a)\to\fdn^{\ocn-1}(\CE,a)$, we see that 
$\ker\vf_\ocn$ is generated by the elements $\tilde\ga^l$, $l\in\zp$.
Then~\er{v1v2ker} follows from~\er{tgalfdq}.

So we have proved that 
the kernel of the homomorphism $\vf_\ocn\cl\fdn^\ocn(\CE,a)\to\fdn^{\ocn-1}(\CE,a)$ 
is contained in the center of the Lie algebra $\fdn^\ocn(\CE,a)$ for any $\ocn\in\zsp$.

Let us prove~\er{hhhhk} by induction on $k$. 
Since $\psi_1=\vf_1$, for $k=1$ property~\er{hhhhk} follows from~\er{v1v2ker}.
Let $n\in\zsp$ be such that~\er{hhhhk} is valid for $k=n$. 
Then for any $h'_1,h'_2,\dots,h'_{n+2}\in\ker\psi_{n+1}$ we have 
\beq
\lb{vfhn2}
\big[\vf_{n+1}(h'_2),\big[\vf_{n+1}(h'_3),\dots,\big[\vf_{n+1}(h'_n),
\big[\vf_{n+1}(h'_{n+1}),\vf_{n+1}(h'_{n+2})\big]\big]\dots\big]\big]=0,
\ee
because $\vf_{n+1}(h'_i)\in\ker\psi_n$ for $i=2,3,\dots,n+2$.
Equation~\er{vfhn2} says that 
\beq
\lb{hhkervf}
\big[h'_2,\big[h'_3,\dots,\big[h'_n,
\big[h'_{n+1},h'_{n+2}\big]\big]\dots\big]\big]\in\ker\vf_{n+1}.
\ee
Since $\ker\vf_{n+1}$ is contained in the center of $\fdn^{n+1}(\CE,a)$,
property~\er{hhkervf} yields 
$$
\big[h'_1,\big[h'_2,\big[h'_3,\dots,\big[h'_n,
\big[h'_{n+1},h'_{n+2}\big]\big]\dots\big]\big]\big]=0.
$$
So we have proved~\er{hhhhk} for $k=n+1$. 
Clearly, property~\er{hhhhk} implies that $\ker\psi_k$ is nilpotent.
\end{proof}

\section*{Acknowledgements}
The author would like to thank M.~Crainic, 
A.~Henriques, I.~S.~Krasil{\cprime}shchik, J.~van de Leur, 
Yu.~I.~Manin, I.~Miklaszewski,  
V.~V.~Sokolov, A.~M.~Verbovetsky, 
and A.~M.~Vinogradov for useful discussions.
 
This work is supported by 
the Netherlands Organisation for Scientific Research (NWO) grants 639.031.515 and 613.000.906.
The author is grateful to the Max Planck Institute for Mathematics (Bonn, Germany) 
for its hospitality and excellent working conditions 
during 02.2006--01.2007 and 06.2010--09.2010, 
when part of this research was done.

\end{document}